\documentclass[twocolumn]{autart}    

\usepackage{graphicx}          
\usepackage[utf8]{inputenc}
\usepackage{amsmath}
\usepackage{amsfonts}
\usepackage{graphicx}
\usepackage{bm}

\usepackage{algorithmic}

\usepackage{algorithm}
\usepackage{afterpage}
\usepackage{wrapfig} 
\usepackage{etoolbox}

\usepackage{xcolor}
\newcommand{\colorname}{black}
\usepackage{mathtools}
\usepackage{dsfont}

\newenvironment{proof}{\paragraph*{Proof:} \vspace{-5mm}}{\hfill$\square$}
\theoremstyle{definition}
\newtheorem{remark}{Remark}[section]
\newtheorem{assumption}{Assumption}

\newtheorem{example}{Example}

\begin{document}

\begin{frontmatter}

\title{Combining Switching Mechanism with Re-Initialization and Anomaly Detection  for Resiliency of Cyber-Physical Systems}\thanks{This work is supported in part by the NSF under grant number CNS \#2039615 and in part by the NYUAD Center for Artificial Intelligence and Robotics, funded by Tamkeen under the NYUAD Research Institute Award CG010.}
\author[NYU]{Hao Fu}\ead{hf881@nyu.edu},    
\author[NYU]{Prashanth Krishnamurthy}\ead{prashanth.krishnamurthy@nyu.edu},               
\author[NYU]{Farshad Khorrami}\ead{khorrami@nyu.edu}  

\address[NYU]{Department of Electrical and Computer Engineering, Tandon School of Engineering,  \\ New York University,  5 MetroTech Center, Brooklyn, New York, 11201}

\begin{keyword}                           
  Cyber-Physical System; Redundancy; Switching Strategy; Re-Initialization; Anomaly Detection; Control System.            
\end{keyword}                             

\begin{abstract}                          
Cyber-physical systems (CPS) play a pivotal role in numerous critical real-world applications that have stringent requirements for safety. To enhance the CPS resiliency against attacks, redundancy can be integrated in real-time controller implementations by designing strategies that switch among multiple controllers. However, existing switching strategies typically overlook remediation measures for compromised controllers, opting instead to simply exclude them. Such a solution reduces the CPS redundancy since only a subset of controllers are used. To address this gap, this work proposes a multi-controller switching strategy with periodic re-initialization  to remove attacks. Controllers that finish re-initialization can be reused by the switching strategy, preserving the CPS redundancy and resiliency. The proposed switching strategy is designed to ensure that at each switching moment, a controller that has just completed re-initialization is available, minimizing the likelihood of compromise. Additionally, the controller's working period decreases with the number of involved controllers, reducing the controller's exposure time to attacks. An anomaly detector is  used to detect CPS attacks during the controller's working period. Upon alarm activation, the current control signal is set to a predefined value, and a switch to an alternative controller occurs at the earliest switching moment. Our switching strategy is shown to be still effective even if the anomaly detector fails to detect (stealthy) attacks.  The efficacy of our strategy is analyzed through three derived conditions under a proposed integrated attack-defense model for mean-square boundedness of the CPS states. Simulation results on a third-order system and a single-machine infinite-bus (SMIB) system confirm that our approach significantly bolsters CPS resiliency by leveraging the advantages of re-initialization, anomaly detection, and switching mechanisms.

\end{abstract}

\end{frontmatter}

\section{Introduction}

Cyber-Physical Systems (CPS) are complex interconnected combinations of heterogeneous hardware and software components through networks. Specifically, CPS integrate sensors, actuators, computational nodes, and physical processes at multiple levels, as illustrated in Fig.~\ref{fig:CPS}.
\begin{figure}
    \centering
    \includegraphics[width=\linewidth]{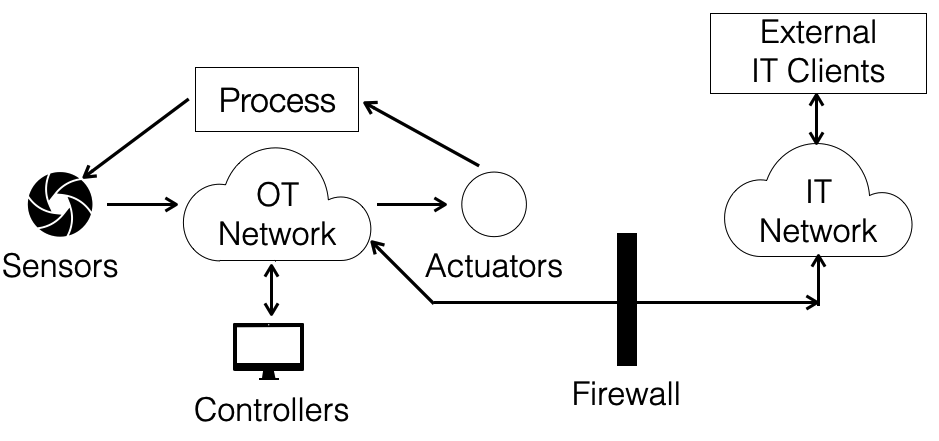}
    \caption{Structure of cyber-physical systems. OT: operational technology. IT: information technology.}
    \label{fig:CPS}
\end{figure}
CPS have been utilized in many critical real-world applications, such as power systems \cite{WK12}, industrial process control \cite{SSTC09}, and water distribution \cite{PMS10}. While advancements in programmability and remote connectivity of computational nodes have eased the operator workload, these developments also render CPS increasingly susceptible to cyber-attacks \cite{KKKPS18}. As illustrated in Fig.~\ref{fig:CPS}, CPS interface with both information technology (IT) and operational technology (OT) networks. This interconnectivity allows attackers to infiltrate the OT network by first breaching the IT network, thereby manipulating data signals. The urgency of this issue is underscored by large increases in incidents reported by the ICS Cyber Emergency Response Team over the past decade (including some high profile attacks) \cite{KKK16}. 

In this paper, we focus on scenarios where attackers have successfully infiltrated the OT network and can manipulate the control signal (we deliberately exclude IT attacks, such as credential misuse and firewall bypassing, as they fall outside the scope of this study).  Various defense mechanisms \cite{KSCKMK16} are available to increase the resiliency of CPS in spite of such threats. For instance, anomaly detectors \cite{CALHHS11,ANR16,HO18,WYPSLW21} can be deployed to identify abnormal activities indicative of an attack. In another approach, the controller can be periodically or dynamically re-initialized to neutralize any ongoing attacks \cite{ACHLMC18,AZKYS19}. Redundancy can also be implemented at different layers of the CPS, such as communication channels, software, or hardware, to serve as backups to primary subsystems \cite{KK20,KK21,ZV21}. This enables the CPS to maintain satisfactory performance levels even when some primary subsystems are compromised or fail. The Moving Target Defense (MTD) strategy \cite{ZDO14} further complicates the attacker's objectives by introducing time-variant and unpredictable elements into the system. For example, a defender could deploy multiple controllers on various computing platforms and operating systems  and could even alter these properties dynamically to maintain a constantly shifting defense posture. 

\begin{figure}
    \centering
    \includegraphics[width=0.9\linewidth]{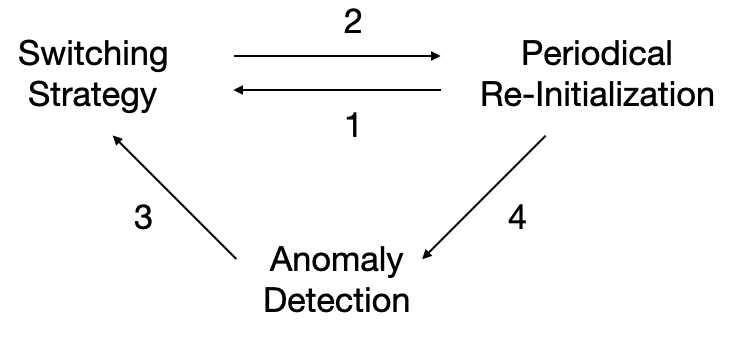}
    \caption{Illustration of how each component of our defense strategy compensates for the limitations of the others. 1: Compromised controllers are re-initialized to remove  attacks and then reintroduced into the switching rotation. 2: The   switching strategy ensures that the control law is in effect, even when the current controller undergoes re-initialization. 3: The anomaly detector alerts the switching strategy to initiate protective measures. 4: Stealthy attacks that evade the anomaly detector are removed through  re-initialization.}
    \label{fig:logic}
\end{figure}

While redundancy, re-initialization, and anomaly detection each offer valuable advantages, they come with their own set of limitations: the control law cannot be applied during the controller's re-initialization period; stealthy attacks might evade anomaly detectors; and redundancy often assumes that only a subset of subsystems will be compromised. Existing strategies \cite{KK20,KK21,ZV21} that integrate redundancy in real-time controller implementations generally switch among multiple controllers to mitigate the impact of an attack on the CPS. These approaches commonly presume that an attacker will only compromise a subset of controllers. For compromised controllers, these methods simply exclude them without remediation measures, thereby reducing the system's overall redundancy. To address this limitation, this paper proposes a novel switching mechanism that incorporates periodic re-initialization of controllers. Rather than being excluded, compromised controllers are re-initialized to remove any attacks and then reintroduced into the switching rotation, thus maintaining system resiliency. Our approach builds upon prior work, wherein controllers are deployed on various computing platforms and operating systems and configured to operate in parallel. The variability \cite{ENK11,LHBF14}—such as changing rotation keys or data locations—can be introduced after each re-initialization, further protecting the controllers from repetitive attacks. The periodic re-initialization and introduced variability allow for a more robust defense strategy, even when all controllers are subject to attack and even when stealthy attacks bypass anomaly detector-based approaches \cite{CALHHS11,ANR16,HO18,WYPSLW21}.

\begin{figure}
    \centering
    \includegraphics[width=\linewidth]{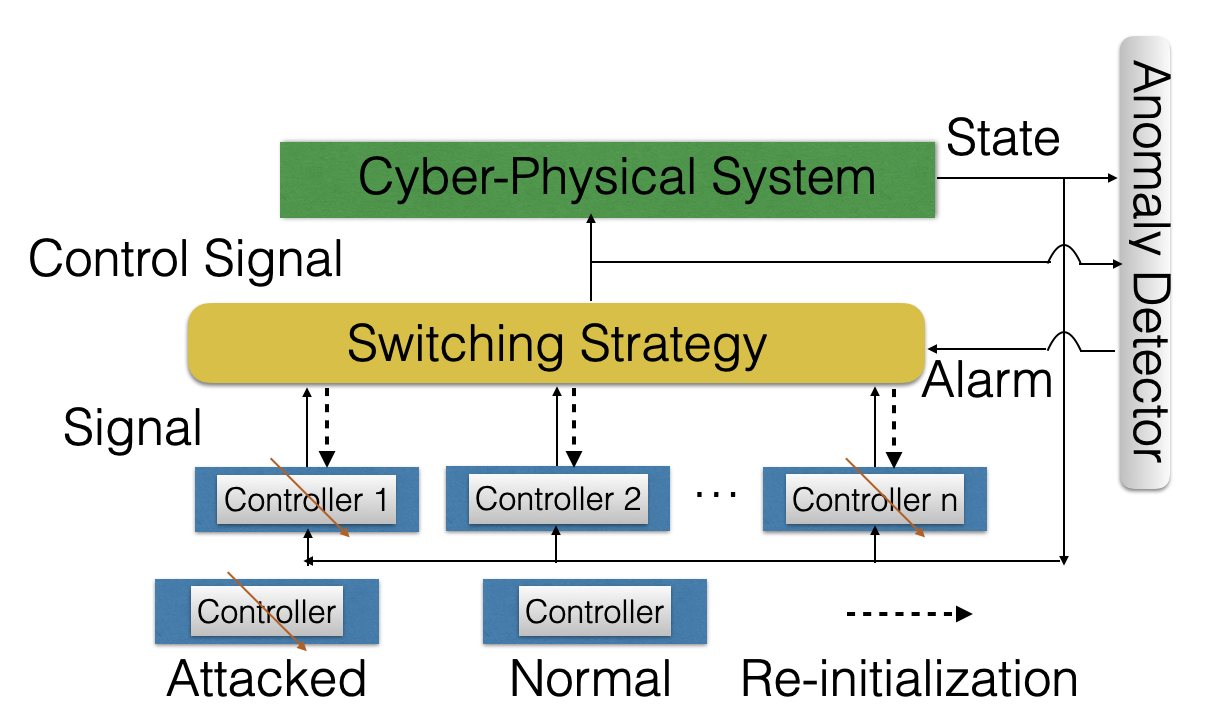}
    \caption{Our approach combines the re-initialization, anomaly detection, and switching defense for the CPS resiliency.}
    \label{fig:RRS2}
\end{figure}

Our switching strategy also mitigates the limitations inherent to the re-initialization approach \cite{ACHLMC18,AZKYS19}. The switching logic is crafted such that at each switching moment, there is always a controller that has just completed the re-initialization process, making it less likely to be compromised at that point in time. This ensures that the control law can continuously be applied to the CPS. Additionally, the period during which each controller is operational—and therefore exposed to potential CPS attacks—decreases as the number of participating controllers increases. An anomaly detector is used to monitor for signs of CPS attacks during each controller's working period. If an alarm is triggered, our switching strategy immediately sets the current control signal to a user-defined value, switches to an alternative controller at the earliest switching moment, and then initiates re-initialization of the replaced controller. Our approach effectively counters stealthy attacks that may elude the anomaly detector since any controller will be swapped out for a fresh one at the next  switching moment and subsequently re-initialized to remove any lingering threats. In other words, if the anomaly detector fails, our switching strategy becomes an effective defense approach that combines switching mechanism with re-initialization. Fig.~\ref{fig:logic} shows how each component of our defense strategy compensates for the limitations of the others, whereas Fig.~\ref{fig:RRS2} illustrates how the defense components are integrated to protect the CPS. Our approach minimizes the likelihood of a successful CPS attack by presenting multiple layers of defense.  Additionally, our approach is designed to be extensible, making it amenable to the inclusion of additional defensive layers, thereby continually improving its protective capabilities.

Our work aligns with similar assumptions and CPS dynamics as those in \cite{KK20,KK21}. For instance, all works consider nonlinear CPS dynamics, a common characteristic in real-world applications like power systems and unmanned vehicles. We also share the assumptions that the overall system is input-to-state stable with respect to external signals, and that the switching mechanism itself is secure against attacks. Additionally, we all consider that the sensor hardware and its transmission to the controller are valid. The two strategies proposed in \cite{KK20,KK21} are designed to counter persistent CPS attacks. Their first strategy  uses each controller for an equivalent duration despite of whether the controller is compromised or not. Their second strategy is  adaptively excludes compromised controllers based on the Lyapunov function values. Our work diverges by assuming all controllers are subject to attacks. Instead of excluding them, we re-initialize and reuse compromised controllers to maintain system redundancy. Additionally, our approach does not hinge on Lyapunov function values for switching. We also demonstrate that their first strategy  serves as a special case of ours. The work of \cite{ZV21} presents a data-based switching strategy for linear, controllable, and observable systems. While their approach also utilizes an anomaly detector, it lacks any form of remediation to recover and reuse compromised controllers. They rely on Lyapunov function values to make switching decisions, a dependency our model avoids. 

To  assess the effectiveness of our approach, this paper proposes an integrated attack-defense model. Sufficient conditions for mean-square boundedness of CPS states are derived under the model 1) when only re-initialization is employed, 2) when re-initialization is combined with anomaly detection, and 3) when our approach in Fig.~\ref{fig:RRS2} is utilized.  We validate our approach through simulations on a third-order system and a single-machine infinite bus (SMIB) system. These results affirm that our approach  bolsters CPS resiliency by inheriting the individual strengths of each defense component, while mitigating their respective limitations. Overall, the contribution of this work includes 1) developing a  switching strategy that is combined with re-initialization and anomaly detection, 2) proposing an integrated attack-defense model for performance analysis, and 3) deriving three mean-square boundedness conditions under the model to analyze the efficacy of our approach. The attack-defense model is proposed in Sec.~\ref{sec:model}. Our switching strategy is introduced in Sec.~\ref{sec:defense}. The conditions for the mean-square boundedness of the CPS states are derived in Sec.~\ref{sec:analysis}.    The simulation results are presented in Sec.~\ref{sec:exp}. The paper is concluded in Sec.~\ref{sec:conclusion}.

\section{CPS, Attack, and Defense Models}
\label{sec:model}

\subsection{The  CPS Dynamics and Its Assumptions}

The considered dynamic nonlinear CPS  has the form:
\begin{align}
    \dot x = f(x,u,w) 
    \label{cps}
\end{align} where $x \in \mathbb{R}^{n_x}$ is the system state, $u \in \mathbb{R}^{n_u}$ is the input, and $w \in \mathbb{R}^{n_w}$ is a disturbance. Additionally, $u$ and $w$ follow physical constraints in the specific CPS (i.e., $|u|\le u_{max}$ and $|w| \le w_{max}$, where $u_{max}$ and $w_{max}$ are two known positive constants). We assume that the CPS has a given control law and Lyapunov function such that:
\begin{assumption}
\label{assumption1}
 Applying the control law $u=u_c(x)$ to the CPS, the given Lyapunov function $V(x)$ satisfies
\begin{align}
    \dot V \le -\alpha(|x|) + \beta_1(|x|)\mu_1(|w|)
    \label{assump3}
\end{align} where $\mu_1(\cdot)$\footnote{The dot  notation, $\cdot$, represents the variable placeholder.} is a class $\mathcal{K}$ function\footnote{Class $\mathcal{K}$ denotes the set of all continuous functions $\alpha: [0,a) \to [0, \infty)$ that are strictly increasing with $\alpha(0)=0$.}, $\beta_1(\cdot)$ is a non-negative function, $\alpha(\cdot)$ is a class $\mathcal{K}_\infty$ function\footnote{Class $\mathcal{K}_\infty$ is the subset of class $\mathcal{K}$ wherein furthermore $a=\infty$ and $\alpha(r)\to \infty$ as $r\to \infty$.}. $\alpha(\cdot)$ satisfies  $\underline \alpha  V(x) \le  \alpha(|x|) \le \overline \alpha V(x)$, $\beta_1(\cdot)$ satisfies  $\beta_1^2(|x|) \le  \overline \beta_1 \alpha(|x|)$, where $\underline \alpha$, $\overline \alpha$, and $\overline{\beta}_1$ are positive constants.
\end{assumption}
\begin{assumption}
\label{assumption2}
The Lyapunov function also satisfies the following inequality with any input $u$:
\begin{align}
    \dot V \le \gamma_1(|x|) + \gamma_2(|x|)\gamma_u(|u|) + \beta_2(|x|)\mu_2(|w|)
    \label{assump4}
\end{align} where $\gamma_2(\cdot)$ and $\beta_2(\cdot)$ are non-negative functions, $\gamma_1(\cdot)$ is any function, and $\gamma_u(\cdot)$ and $\mu_2(\cdot)$ are class $\mathcal{K}$ functions. Additionally,  $\gamma_1(\cdot)$ satisfies $\gamma_1(|x|) \le \overline{\gamma}_1 V(x)$, $\gamma_2(\cdot)$ satisfies $\gamma_2^2(|x|) \le \overline{\gamma}_2 V(x)$, and $\beta_2(\cdot)$ satisfies $\beta_2^2(|x|) \le \overline{\beta}_2 V(x)$, where $\overline{\gamma}_1$ is a constant number, and $\overline{\gamma}_2$ and $\overline{\beta}_2$ are positive constant numbers.
\end{assumption}
\noindent
\begin{remark}
\label{remark1}
The underlying process dynamics of several CPS (e.g.,  power systems and unmanned vehicles) can be represented in \eqref{cps}. The strongest part of {\bf Assumption~\ref{assumption1}} is that the negative term $-\alpha(|x|)$ in \eqref{assump3} is of the “same size” (in a nonlinear function sense) as the Lyapunov function $V(x)$. {\bf Assumption~\ref{assumption2}} considers a worst-case scenario when the attacker arbitrarily modifies the control input signal to cause the maximal adversarial impact in terms of a Lyapunov inequality. {\bf Assumptions~\ref{assumption1}} and {\bf \ref{assumption2}}  are standard assumptions for nonlinear control systems, stating that a controller has been designed to render the overall system input-to-state stable with respect to exogenous signals. They ensure that the CPS \eqref{cps} will not have a finite escape time.  They are analogous to the conditions considered in \cite{KK20,KK21} and hold under various dynamic systems and control design approaches such as backstepping-based control designs for strict-feedback systems, feedback linearization, and dynamic high-gain-based control designs. The following are two examples.
\end{remark}  

\begin{example}
\label{exm:backstep}
Consider a backstepping-based control design for a general strict-feedback system with the form
\begin{align}
\nonumber
    \dot x_i &= f_i(x_1,..., x_i)+\phi_i(x_1,...,x_i)x_{i+1}, ~  i=1,...,n-1 \\
    \dot x_n &= f_n(x_1,...,x_n)+\phi_n(x_1,...,x_n)u
\end{align} where $f_1,...,f_n,\phi_1,...,\phi_n$ are certain functions. In the $i$th step with $i<n$, a virtual control law $x_{i+1}^\prime$ is designed. Define $z_{i+1} = x_{i+1} - x_{i+1}^\prime$ as the difference between $x_{i+1}$ and $x_{i+1}^\prime$. In the $n$th step, the real control law $u$ is designed. It can be seen that the corresponding Lypanov function $V$ has quadratics in $z_1,...,z_n$ and $\dot V$ automaticaly satisfies {\bf Assumptions~\ref{assumption1} and \ref{assumption2}}.
\end{example}

\begin{example}
    \label{exm:linear}
    Consider a linear system with $\dot x = Ax+Bu+Hw$ and a feedback control law $u=Kx$, where $K$ is a static control gain. Define $V=x^TPx$, where $P$ is a symmetric positive-definite matrix such that $P(A+BK)+(A+BK)^TP\le -I$ with $I$ being the identity matrix, then we have
    \begin{align}
    \nonumber
        \dot V & \le - x^Tx+2x^TPHw \le -\alpha(|x|) + \beta_1(|x|)\mu_1(|w|) \\
        \nonumber
        \dot V & \le x^T(PA+A^TP)x+2x^TPBu+2x^TPHw \\
        &\le \gamma_1(|x|) + \gamma_2(|x|)\gamma_u(|u|) + \beta_2(|x|)\mu_2(|w|)
    \end{align} for some $\alpha(\cdot)$, $\mu_1(\cdot)$, $\beta_1(\cdot)$, $\gamma_1(\cdot)$, $\gamma_2(\cdot)$, $\gamma_u(\cdot)$, $\beta_2(\cdot)$, and $\mu_2(\cdot)$ that satisfy {\bf Assumptions~\ref{assumption1} and \ref{assumption2}}. 
\end{example}

{\bf Inertia}: CPS inherently possess the characteristic of gradual state evolution,  referred to as inertia, which results in delayed reactions to external changes. Emulating this inertia property in traditional systems usually incurs additional memory and/or hardware costs \cite{AZKYS19}.

\subsection{The Adversary and Attack Model}

{\bf Goal}: The  attacker's goal is to thwart the intended behavior of the CPS by manipulating the control signal to an arbitrary value bounded by $u_{max}$, subject to the physical constraints inherent to the specific CPS.  \cite{CALHHS11} demonstrates that min/max attacks are the most effective in this context. Accordingly, we assume that the attacker sets the control signal to $u_{max}$. \textcolor{\colorname} {The overall control signal is given by
\begin{align}
    u(x,t) = \begin{cases} 
u_c(x) & \text{normal,}  \\
u_{max} & \text{under attacks}.
\end{cases}
\label{eq:attack}
\end{align} where $u_c(x)$ is the control law. Consequently, the evolution of the CPS is given by
\begin{align}
      \dot x =  \begin{cases} 
f(x,u_c(x),w)  & \text{normal,}  \\
f(x,u_{max},w)  & \text{under attacks}.
\end{cases}
\end{align} Different systems lead to various forms of the control laws and attack models. As an example, the overall system in {\bf Example~\ref{exm:linear}} under attacks is given by
\begin{align}
     \dot x =  \begin{cases} 
     (A+BK)x+Hw  & \text{normal,} \\
     Ax+Bu_{max}+Hw  & \text{under attacks.} 
\end{cases}
\end{align} }

{\bf Ability}: \textcolor{\colorname}{ The attacker could utilize various methods (e.g., software bugs, hardware faults) to infiltrate the controller nodes and initiate the manipulation of the control signals to cause dynamic impact to the CPS.} To achieve their objective, we assume the attacker has comprehensive knowledge of the CPS internals. This level of understanding enables them to exploit vulnerabilities in the CPS's software and network infrastructure. As a result, the attacker gains access to the controller node, allowing them to modify the control signal. The attacker can spoof the sensor value in the controller implementation.  However, the sensor hardware and its transmission to the controller are considered valid. Additionally, the proposed switching mechanism  is assumed to remain secure from such compromises.  \textcolor{\colorname}{The defender can implement various safeguards to attempt to shield the switching mechanism from potential attacks. This includes, for example, segregation of the network, implementation using an analog switch in conjunction with a bare-metal embedded processor, etc. As such, the attack model considered is that the switching mechanism is trustworthy. }

{\bf Attack Time}:
In our model, we take into account the time needed for the attacker to successfully exploit the CPS's software and network vulnerabilities. In a worst-case scenario, the attacker might need zero time to hijack the control flow. However, due to the CPS's inherent inertia property, the system takes time to respond to the attacker's manipulations. Further constraints might arise from limitations in the attacker's equipment or other environmental factors, making instantaneous hijacking unlikely. Moreover, if certain defense measures, such as system diversification \cite{LHBF14}, are in place, the attacker will require additional time to breach the system. Detailed examples concerning exploit time are elaborated in \cite{AZKYS19}.  We model this attack time as $t_a^\prime$. For defense strategies that are probabilistic in nature \cite{ENK11,LHBF14}, this attack time also becomes a stochastic variable. Based on these considerations, we make the following assumption:
\begin{assumption}
\label{assumption4}
The time required for the attacker to launch an attack and cause the CPS to respond, denoted as $t_a^\prime$, is a random variable.
\end{assumption}

\textcolor{\colorname}{Consider the moment that the attacker starts launching the CPS attack as $t=0$, then the control signal with time can be written as
\begin{align}
    u(x, t|~t_a^\prime) = \begin{cases} 
u_c(x) & t<=t_a^\prime,  \\
u_{max} & t>t_a^\prime.
\end{cases}
\end{align}
Similarly, the CPS dynamics can be written as
\begin{align}
\label{eq:cps_attack}
    \dot x = f(x, u(x, t|~t_a^\prime), w) = \begin{cases}
        f(x, u_c(x), w) & t\le t_a^\prime, \\
        f(x, u_{max}, w) & t>t_a^\prime.
    \end{cases}
\end{align}
The attacker's goal can be written as
\begin{align}
\label{eq:goal_vanila}
   \text{Find } ~t_a^\prime~~ \text{   s.t. } \lim_{t\to\infty}\mathbb{E}[V(t)] \to \infty
\end{align} where $V(t)$ is the given Lypanov function in Assumptions~\ref{assumption1} and \ref{assumption2} and the CPS dynamics are given in \eqref{eq:cps_attack}. 
}

\textcolor{\colorname}{
{\bf Distribution of $t_a$}: The distribution information of $t_a$ in this work is only used in the analysis of the stability of the attacked system. Nevertheless, if the defender wishes to estimate the distribution information of $t_a$, there are a few possible methods that could be applied, including 1) analyzing historical incidents, documents, and reports to estimate $t_a$’s distribution, 2) applying stochastic models (e.g., ones  as summarized in Table II of \cite{A21}), and 3) treating $t_a$ as a constant to find the minimum value  that meets the mean-square boundedness conditions dervied in the following section, then considering only attacks whose $t_a $ is less than the derived minimum value.
}

{\bf Impact of the CPS Attack under Assumption~\ref{assumption4}}: Without any resilient control strategy, the attacker's maximum impact under {\bf Assumption~\ref{assumption4}} is modeled via the Lyapunov inequality given in {\bf Assumption~\ref{assumption2}}, \textcolor{\colorname}{ which meets the attacker's goal in \eqref{eq:goal_vanila}.} A more general analysis is provided in Sec.~\ref{sec:analysis}. {\bf Assumption~\ref{assumption4}} is used only for performance analysis. Implementing our defense does not require this assumption.

{\bf Faults}: 
Attacks and faults differ in origin: attacks are intentional actions aimed at disrupting system performance, while faults are naturally occurring anomalies. Detecting attacks on a CPS is often more challenging than identifying faults, as a savvy attacker may know how to evade detection mechanisms \cite{CALHHS11}. Specifically, \textcolor{\colorname}{fault-tolerant systems do not account for stealthiness, whereas this work  considers stealthy attacks that evade detection during the stability analysis.} However, both attacks and faults can produce similar effects on a system, making them difficult to distinguish without post-event analysis.  Mathematically, our proposed threat model, along with Assumption~\ref{assumption4}, is applicable to both attacks and faults. Consequently, similar to many other works \cite{ACHLMC18,MCMK19,KK20,KK21}, we treat attacks and faults interchangeably in our analysis.  \textcolor{\colorname}{The attacker could utilize various methods (e.g., software bugs and hardware faults) to gain access to any controller nodes and initiate such manipulations of the control signals to cause dynamic impact to the CPS.} It is noted that the system disturbance  $w$ in \eqref{cps} is distinct from faults or attacks  since it affects the CPS as a whole. 

\subsection{The Defender and Attack-Defense Model}

{\bf Goal}: The defender's goal is to design a multi-controller switching strategy that enables the CPS to sustain acceptable performance levels despite the presence of bugs, failures,  faults, and attacks, provided these issues do not breach the CPS's physical constraints.

{\bf Ability}: The defender can take necessary measures to ensure that the switching mechanism  is well protected from attacks. For instance, they can employ an analog switch paired with a bare-metal embedded processor to enact the switching law with no external ports but only one-way communication. Controllers can be instantiated on different computers or operating systems. The defender possesses complete knowledge of the CPS.

{\bf Re-initialization}: 
The re-initialization defense strategy in this work periodically resets the controllers and introduces  variability after each reset to enhance data protection. By doing so, the attack window available to a potential intruder is bounded, constraining the sustained impact of any attack. This variability also makes it more challenging for an attacker to succeed.  Please refer to \cite{AZKYS19} for additional background on the benefits of re-initialization. In our approach, controllers are re-initialized at regular intervals, specifically every  $T_0$ time units, where $T_0$ is a defender-determined constant. During this re-initialization period, the control signal can be set to any value determined by the defender. For consistency, we opt for setting the control signal to zero during re-initialization, which is adopted by  \cite{ACHLMC18}. The time $t_r$ required to complete a re-initialization is hardware-dependent and considered constant. \textcolor{\colorname}{ To address potential issues (such as oscillations) that could happen during controller restarting, during which time controllers might need to re-learn operational states, one can take the potential solutions, including 1) using static instead of dynamic controllers, 2) deploying multiple controllers to ensure continuous operation, even during restarts, 3) optimizing controller implementation to enable faster reboots \cite{GSYGA04}, and 4) restoring controllers to their previous operational states using checkpoints \cite{KXWSL18}. This work considers both the first approach and a combination of the second and fourth approaches by re-initializing and recovering the compromised controllers. The third approach is an application-specific implementation detail.
}

{\bf Implementation Challenges}: To streamline our analysis, we overlook potential time delays within the OT network, for example, those occurring between sensors, control systems, and actuators. These delays are indirectly accounted for via built-in authentication times. Specifically, we designate a certain time, $t_c$, required for the CPS to authenticate a  controller.  $t_c$ is considered as a constant dependent on the CPS and the controller.

\begin{figure}[ht]
    \centering
    \includegraphics[width=\linewidth]{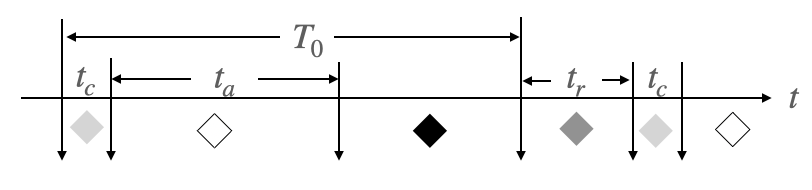}
    \caption{The sequence of events under the re-initialization defense. $t_a=\min\{t_a^\prime, T_0-t_c\}$. White: The CPS works normally. Dark: The CPS attack successfully compromises the controller. Dark gray: The controller is being re-initialized. Shallow gray: The CPS is authenticating the controller.}
    \label{fig:a-d-model}
\end{figure}

Fig.~\ref{fig:a-d-model} illustrates the sequence of events when applying the re-initialization defense strategy, where $t_a=\min\{t_a^\prime, T_0-t_c\}$  because the controller is re-initialized anyway after being active for $T_0$. From the beginning to time $t_c$, the CPS is authenticating the controller during which $u\equiv 0$ (since the control law cannot be applied).  At $t_c$, the CPS and controller successfully establish the connection. The controller outputs correct control signals $u=u_c(x)$ from $t_c$ to $t_a$. At $t_a$,  the attacker successfully compromises the controller and alters the signal to $u=u_{max}$. The attack exists $T_0-t_a-t_c$ time until a re-initialization starts. The re-initialization takes a time of $t_r$. During the re-initialization period,  $u\equiv 0$. \textcolor{\colorname}{
The overall control signal with $T_0$, $t_r$, and $t_c$ can be written as
\begin{align}
\nonumber
    u=\begin{cases}
        0 & k(T_0+t_r)\le t < k(T_0+t_r)+t_c,  \\
        u_c(x) & k(T_0+t_r)+t_c \le t < k(T_0+t_r)+t_a, \\
        u_{max} & k(T_0+t_r)+t_a \le t< k(T_0+t_r)+T_0, \\
        0 & k(T_0+t_r)+T_0 \le t< (k+1)(T_0+t_r).
    \end{cases}
\end{align} where $u=u(x, t|t_a, T_0, t_c, t_r)$ and $k=\lfloor t/(T_0+t_r) \rfloor$. Similarly, the CPS dynamics can be written as
\begin{align}
\label{eq:cps_reboot}
    \dot x &= f(x, u(x, t|~t_a, T_0, t_c, t_r), w) = \\
    \nonumber
    & \begin{cases}
        f(x, 0, w) & k(T_0+t_r)\le t < k(T_0+t_r)+t_c,  \\
       f(x, u_c(x), w)  & k(T_0+t_r)+t_c \le t < k(T_0+t_r)+t_a, \\
      f(x, u_{max}, w)   & k(T_0+t_r)+t_a \le t< k(T_0+t_r)+T_0, \\
      f(x, 0, w)   & k(T_0+t_r)+T_0 \le t< (k+1)(T_0+t_r).
    \end{cases}
\end{align}
The defender's goal under the re-initialization strategy can be written as
\begin{align}
     \text{Find } ~T_0~\&~t_r~~ \text{   s.t. } \lim_{t\to\infty}\mathbb{E}[V(t)] < \infty
\end{align} where the  system dynamics are given by \eqref{eq:cps_reboot}.
}

{\bf Anomaly Detector}: The defender could combine the re-initialization defense with anomaly detection, as shown in Fig.~\ref{fig:anomalydefense}.
\begin{figure}
    \centering
    \includegraphics[width=\linewidth]{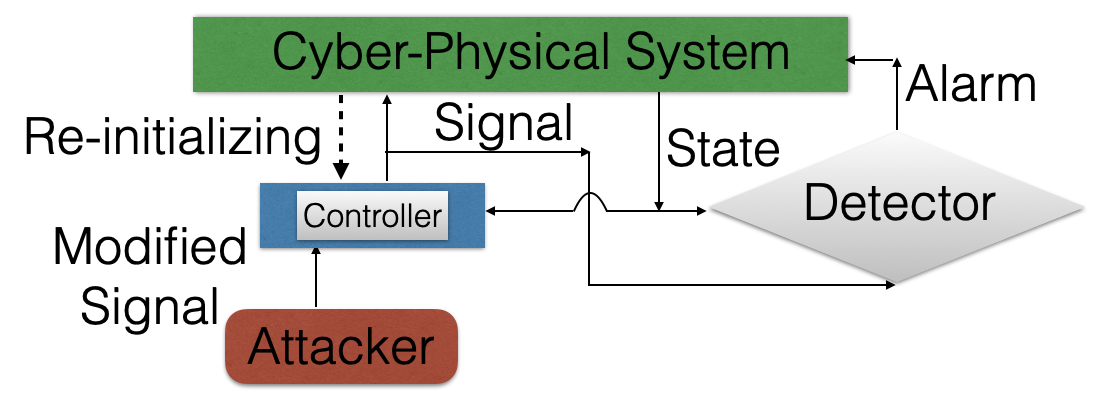}
    \caption{Combining re-initialization with anomaly detection.}
    \label{fig:anomalydefense}
\end{figure}
In this study, we adopt change detection techniques \cite{CALHHS11} that perform hypothesis testing based on historical observations of the CPS state. Specifically, the anomaly detector calculates the expected behavior of the CPS using the knowledge of CPS internals. If observed behavior deviates from the expected, the detector rejects the null hypothesis $H_0$—that the system is operating normally—in favor of the alternative hypothesis $H_1$—that an attack is occurring.  The defender predetermines the alarm threshold for hypothesis rejection, allowing for a trade-off between detection sensitivity and false alarm rate. Importantly, even if an attack goes undetected due to a high threshold, the controller will undergo re-initialization after  $T_0$ time units, mitigating the attack's impact. Though we do not explicitly consider false alarms, our approach is inherently robust against them. False alarms can be viewed as merely increasing the effective attack rate, and their impact on system performance is indistinguishable from actual attacks. Regarding detection timing, we recognize that the anomaly detector needs time $t_d^\prime$ to accumulate sufficient data to make a decision, meaning that the detection process is running in  real-time but the actual detection/alarm could happen with a delay. Given that the detection relies on statistical methods, we make the following assumption that accounts for the probabilistic nature of detection time to enhance the model's realism:
\begin{assumption}
\label{assumption3}
The time  $t_d^\prime$ required for the anomaly detector to identify a CPS attack is a random variable.
\end{assumption}

Fig.~\ref{fig:rcp} shows the defense after incorporating the anomaly detector, where $t_d=\min\{t_d^\prime, T_0-t_a-t_c\}$ represents the duration that the controller is compromised. The reason to use $t_d$ instead of $t_d^\prime$ is that the controller will be re-initialized after   $T_0$ regardless of attack status. The value of $t_d$ is determined by both $t_d^\prime$ and $t_a$. Empirical results from \cite{CALHHS11} empirically indicate that the detection time $t_d^\prime$ varies with different characteristics of attacks.  As such, it is reasonable to consider that $t_d$ could also be subject to variations depending on the specific nature of the attack.

\begin{figure}[ht]
    \centering
    \includegraphics[width=0.45\textwidth]{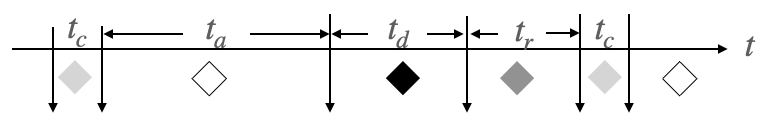}
    \caption{The sequence of events under the anomaly detector. $t_d=\min\{t_d^\prime, T_0-t_a - t_c\}$. White: The CPS works normally. Dark: The CPS attack successfully compromises the controller. Dark gray: the anomaly detector detects the attack and the controller is being re-initialized. Shallow gray: The CPS is authenticating the  controller.}
    \label{fig:rcp}
\end{figure}

\textcolor{\colorname}{
The control signal in Fig.~\ref{fig:rcp} can be written as
\begin{align}
\nonumber
    u = \begin{cases}
        0 &  T^\prime \le t < T^\prime+ t_c,  \\
        u_c(x) & T^\prime+t_c \le t < T^\prime+t_c+t_a, \\
        u_{max} & T^\prime+t_c+t_a \le t< T^\prime+t_c+t_a+t_d, \\
        0 & T^\prime+t_c+t_a+t_d \le T^\prime+t_c+t_a+t_d+t_r,
    \end{cases}
\end{align} where $u=u(x, t |~ t_a, t_d, t_c, t_r, T_0)$ and  $T^\prime$ represents the duration from $t=0$ to the completion of last re-initialization. The CPS dynamics are given as 
\begin{align}
\label{eq:cps_anomaly}
    \dot x = f(x, u(x, t |~t_a, t_d, t_c, t_r, T_0), w).
\end{align}
The defender's goal can be written as
\begin{align}
     \text{Find } ~T_0~,~t_r,~\&~t_d^\prime~~ \text{   s.t. } \lim_{t\to\infty}\mathbb{E}[V(t)] < \infty
\end{align} where the  system dynamics are given by \eqref{eq:cps_anomaly}.
}

\begin{remark}
\label{remark2}
Many works  \cite{A21,XH19,RLNW16,LHKK18,MCS13,GUCVFRTSC18} have been considering probabilistic attacks and/or detection. This work aligns with the same considerations.   When involving stochastic components, many works \cite{TS12,AESE21,HMS03,KM01} study the mean-square stability and boundedness of system states. Therefore, this paper also considers the mean-square boundedness of the CPS states, i.e., $\lim_{t\to\infty}\mathbb{E}[V(t)]<+\infty$.
\end{remark}

\section{The Proposed Switching Strategy}
\label{sec:defense}

We introduce a multi-controller switching strategy to integrate redundancy into real-time controller deployments to mitigate the impact of attacks on the CPS. The switching strategy is designed to work in conjunction with re-initialization and anomaly detection, thereby leveraging the strengths of each approach. The structure of this integrated defense mechanism is depicted in Fig.~\ref{fig:RRS2}, while its implementation is outlined in Alg.~\ref{alg:RRS2}.

\begin{algorithm} 
 \caption{Switching Strategy in Fig.~\ref{fig:RRS2}}  
 \label{alg:RRS2} 
 \begin{algorithmic} 
     \STATE Given system $\mathcal{S}$, an anomaly detector $\mathcal{N}$, $n$ controller instantiations $u_{i}(x)$ with $i=1,...,n$, re-initialization time $t_r$, authentication time $t_c$, working time $T_0 = \frac{1}{n-1}t_r$ for each controller, initial controller $i^* = 1$
     \WHILE{$\mathcal{S}$ is working}
      \FOR{t =  $0$ $\to$ $t_c$}
      \STATE Authenticating the controller $i^*$
      \ENDFOR
       \FOR{t = $t_c$ $\to$ $T_0$}
        \IF{$\mathcal{N}(\mathcal{S})$ generates alarm}
        \STATE Apply $u\equiv 0$ to the system $\mathcal{S}$
        \ELSE
        \STATE Apply the control law $u_{i^*}(x)$ to the system $\mathcal{S}$
        \ENDIF
       \ENDFOR
       
       \STATE Re-initialize controller $i^*$
       
       \STATE $i^* = (i^* \text{ mod } n) + 1 $
	\ENDWHILE
 \end{algorithmic}
 \end{algorithm}

Alg.~\ref{alg:RRS2} begins by computing the operational duration $T_0$ for each controller. This ensures that when the currently active controller reaches the end of its $T_0$ interval, a substitute controller has completed its re-initialization and is ready for a switch. The algorithm then replaces the active controller with this prepped alternative and commences re-initializing the former.   When there are  $n>1$ controllers available, $T_0$ is calculated as $\frac{t_r}{n-1}$. To intuitively understand this choice of $T_0$, consider $n=2$. When the CPS switches to Controller 2, Controller 1 starts its re-initialization process. At the end of the $T_0$ interval, Controller 2 begins re-initialization, while Controller 1 completes its own and becomes available. The active time of each controller—and hence its exposure to potential attacks—decreases as the number of controllers increases. The anomaly detector functions to identify attacks during each controller's active period. Should an alarm be triggered, Alg.~\ref{alg:RRS2} nullifies the control signal  until the next scheduled switch. The attacker can target any subset of the $n$ controllers at the same time. \textcolor{\colorname}{ The  CPS dynamics with the given strategy are
\begin{align}
\label{eq:cps_switch}
    \dot x = f(x, \mathcal{U}(t, x, \{u_i\}_{i=1}^n), w)
\end{align} where $\mathcal{U}(t, x, \{u_i\}_{i=1}^n) = u_{i^*}(x, t|~t_a, t_d, T_0, t_r, t_c)$ with $i^*=\lfloor t/T_0 \rfloor \mod n + 1$. The control signal of controller $i^*$ is given as $u_{i^*}(x, t|~t_a, t_d, T_0, t_r, t_c)=$
\begin{align}
\nonumber
  \begin{cases} 
  0 & 0 \le t < (i^*-1)T_0 \\
     0 & G(i^*, t) \le t < G(i^*, t) +t_c, \\
     u_c(x) &  G(i^*, t) +t_c \le t <G(i^*, t) +t_c+t_a \\
     u_{max} & G(i^*, t) +t_c+t_a  \le t <  G(i^*, t) +t_c+t_a +t_d \\
     0 & G(i^*, t)+t_c+t_a +t_d \le t  < G(i^*, t)+T_0+t_r
\end{cases}
\end{align} where $G(i^*, t) = (i^*-1)T_0+\lfloor (t-(i^*-1)T_0)/(T_0+t_r) \rfloor(T_0+t_r)$.
The defender's goal can be written as
\begin{align}
     \text{Find } ~n~,~t_r,~\&~t_d^\prime~~ \text{   s.t. } \lim_{t\to\infty}\mathbb{E}[V(t)] < \infty
\end{align} where the  system dynamics are given by \eqref{eq:cps_switch}.
}

\begin{figure}
    \centering
    \includegraphics[width=0.9\linewidth]{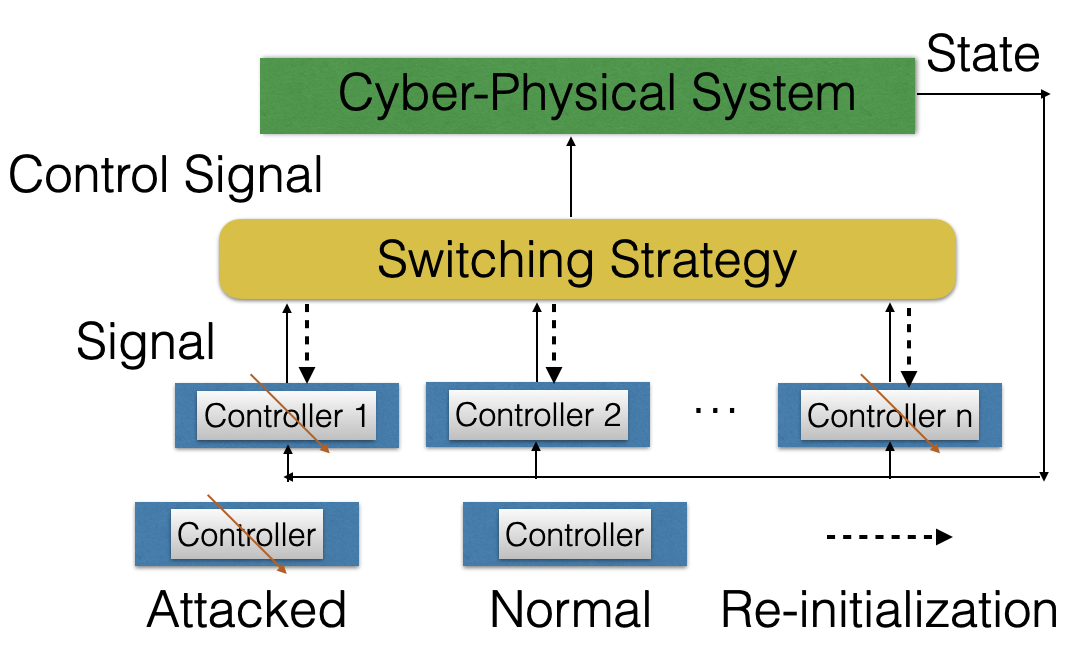}
    \caption{Alg.~\ref{alg:RRS2} without the anomaly detector.}
    \label{fig:alg2}
\end{figure}

Alg.~\ref{alg:RRS2} is designed to ensure the continuous application of the control law to the CPS, even when the current controller undergoes re-initialization. By integrating an anomaly detector, the algorithm further mitigates the adverse effects of attacks on the operational controller. Stealthy attacks that  evade the detector can still be eliminated through periodic re-initialization. Although the algorithm amalgamates the strengths of re-initialization, anomaly detection, and a controller-switching mechanism, \textcolor{\colorname}{the implementation feasibility} should be discussed. 

\textcolor{\colorname}{
{\bf Implementation Feasibility}: Re-initialization is modeled as a reset of the controller hardware/software depending on its implementation in a particular cyber-physical system. At re-initialization, the controller’s memory is considered to be wiped out and the initial good image of the controller’s software is loaded from a read-only storage to achieve an operational state. Since a re-initialization essentially resets the controller state, it is possible that there could be transients (e.g., oscillations) for a time period after re-initialization depending on the controller structure.  Static controllers would not exhibit such transients since their control signal commands are dependent only on the current CPS states. However, for dynamic controllers, setting the controller to its initial state may cause transients such as oscillations due to the loss of their previous states. Methods such as setting the controller to a recent checkpoint upon re-initialization could be applied to mitigate this issue \cite{KXWSL18}. One additional point to note is that while the read-only storage can typically be considered trustworthy in many cyber-physical system applications, the proposed approach can be applied even under the scenario where there is a possibility that the stored images are compromised (e.g., if the adversary has managed to compromise the stored image when the anomaly detector took a long time to detect the attack or even failed to detect the attack). Loading a compromised image is equivalent to the zero-time hijack with $t_a=0$. As noted in the following section, then the CPS can still be shown to be mean-square bounded if the number of such compromised controllers is less than a derived threshold.
}

\textcolor{\colorname}{
Several measures can be taken to mitigate the potential dynamic impacts/transients during re-initialization, such as using an improved anomaly detector to reduce detection delay and increase true alarm rates, using static controllers, and estimating/storing recent states of the controller for use after re-initialization. For the last measure, one could checkpoint the system states and controller states. Once an alarm is initiated, a state estimator could use these checkpoints to calculate the correct updated states for controllers. Considering the detection delay, a sliding window can be used to ensure that the checkpointed values are trustworthy. Such methods have been discussed in the literature \cite{KXWSL18,ZCKC20}. 
}

We also acknowledge that as the number of involved controllers increases, practical challenges arise in terms of their design, implementation in real-world embedded systems, and subsequent maintenance.  Additionally, the incorporation of an anomaly detector entails additional expenditures for design, installation, implementation, and maintenance. In scenarios where an anomaly detector fails to detect (stealthy) attacks or is not available due to these costs, the defense described by Fig.~\ref{fig:RRS2} becomes the one in Fig.~\ref{fig:alg2}. However, our switching strategy in this degraded mode still remains effective. Our analysis demonstrates that both switching algorithms in Fig.~\ref{fig:RRS2} and Fig.~\ref{fig:alg2} prove effective when a sufficient number $n$ of controllers are involved.

\section{Analysis of the Attacks and Defenses}
\label{sec:analysis}

\subsection{Performance of the Defense in Fig.~\ref{fig:a-d-model}}

\begin{thm}
\label{theorem2}
Under {\bf Assumptions~\ref{assumption1}$-$\ref{assumption4}}  and the scenario in Fig.~\ref{fig:a-d-model}, denote $t_a = \min\{t_a^\prime, T_0-t_c\}$, a sufficient condition for the CPS states in \eqref{cps} to be mean-square bounded is 
\begin{align}
    \mathbb{E}_{t_a}[e^{-(\lambda+\lambda_a) {t_a}+\lambda_a(T_0 +t_r)}] < 1
    \label{eq:reboot}
\end{align} where $\lambda = \underline{\alpha} - \frac{1}{\epsilon}\overline{\beta}_1\overline{\alpha}$ with $\epsilon$ being any positive constant that satisfies $\epsilon > \frac{1}{\underline{\alpha}}\overline{\beta}_1\overline{\alpha}$ and $\lambda_a = \overline{\gamma}_1 + \frac{1}{\epsilon_a}\overline{\gamma}_2+\frac{1}{\epsilon_b}\overline{\beta}_2$ with $\epsilon_a$ and $\epsilon_b$ being any positive constants.
\end{thm}
\begin{proof}
Without loss of generality, assume that the authentication starts at $t_0$, then at $t_0+t_c$,  the CPS follows  (\ref{assump4}) with $u\equiv 0$. After another $t_a$, from $[t_0+t_c, t_0+t_c+t_a]$, the CPS follows  (\ref{assump3}) with $u=u_c(x)$. During $[t_0+t_c+t_a, t_0+T_0]$, the CPS follows (\ref{assump4}) with changed control signal $u=u_{max}$. During $[t_0+T_0, t_0+T_0+t_r]$, the CPS follows  (\ref{assump4}) with $u\equiv 0$. Specifically, from $t_0$ to $t_0+t_c$ with $u\equiv 0$, using the inequality $ab\le (a^2+b^2)/2$, we have 
\begin{align}
    \dot V  & \le (\overline{\gamma}_1 + \frac{1}{\epsilon_a}\overline{\gamma}_2+\frac{1}{\epsilon_b}\overline{\beta}_2)V + \frac{\epsilon_b}{4}\mu_2^2(|w|).
\end{align}
Using Bellman-Gr\" onwall's inequality \cite{G19,B43}, we have
\begin{align}
    &V(t_0+t_c) \le 
     V(t_0)e^{\lambda_a t_c} + \overline{w}_0
\end{align} where $\overline{w}_0 = \frac{\epsilon_b}{4}\mu_2^2(|w|)\int_{t_0}^{t_0+t_c}e^{\lambda_a (t_0+t_c-\tau)}d\tau$.
From $t_0+t_c$ to $t_0+t_c+t_a$, we have 
\begin{align}
\label{eq:ineq}
    \dot V &\le -\alpha(|x|) + \frac{\beta_1(|x|)}{\sqrt{\epsilon/2}}\mu_1(|w|)\sqrt{\epsilon/2} \\
    & \le -\alpha(|x|) + \frac{\beta_1^2(|x|)}{\epsilon} + \frac{\epsilon}{4}\mu_1^2(|w|) \\
    &\le  -(\underline{\alpha} - \frac{1}{\epsilon}\overline{\beta}_1\overline{\alpha}) V + \frac{\epsilon}{4}\mu^2_1(|w|).
\end{align}
According to Bellman-Gr\" onwall's inequality, 
\begin{align}
    V(t_0+t_c+t_a) 
     \le V(t_0+t_c)e^{-\lambda t_a} + \overline{w}_1
\end{align} where  $\overline{w}_1 = \frac{\epsilon}{4}\mu_1^2(|w|)\int_{t_0+t_c}^{t_0+t_c+t_a} e^{-\lambda (t_0+t_c+t_a - \tau)} d\tau$.
From $t_0+t_c+t_a$ to $t_0+T_0$, similar to \eqref{eq:ineq}, we have
\begin{align}
\nonumber
    \dot V &
     \nonumber
    &\le (\overline{\gamma}_1 + \frac{1}{\epsilon_a}\overline{\gamma}_2+\frac{1}{\epsilon_b}\overline{\beta}_2)V + \frac{\epsilon_a}{4}\gamma_u^2(|u|)+\frac{\epsilon_b}{4}\mu_2^2(|w|).
\end{align}
Similarly, according to Bellman-Gr\" onwall's inequality:
\begin{align}
\nonumber
    &V(t_0+T_0) \le 
      V(t_0+t_c+t_a)e^{\lambda_a (T_0-t_a-t_c)} + \overline{u} + \overline{w}_2
\end{align} where  
$\overline{u} = \frac{\epsilon_a}{4}\gamma_u^2(|u|)  \int_{t_0+t_c+t_a}^{t_0+T_0} e^{\lambda_a (t_0+T_0-\tau)}  d\tau$ and $\overline{w}_2 =\frac{\epsilon_b}{4}\mu_2^2(|w|)  \int_{t_0+t_c+t_a}^{t_0+T_0}e^{\lambda_a (t_0+T_0-\tau)} d\tau$.
Since $u\equiv 0$ during $t_0+T_0$ to $t_0+T_0+t_r$, we have
\begin{align}
    \dot V  & \le (\overline{\gamma}_1 + \frac{1}{\epsilon_a}\overline{\gamma}_2+\frac{1}{\epsilon_b}\overline{\beta}_2)V + \frac{\epsilon_b}{4}\mu_2^2(|w|).
\end{align}
Using Bellman-Gr\" onwall's inequality, we have
\begin{align}
    &V(t_0+T_0+t_r) \le 
     V(t_0+T_0)e^{\lambda_a t_r} + \overline{w}_3
\end{align} where $\overline{w}_3 = \frac{\epsilon_b}{4}\mu_2^2(|w|)\int_{t_0+T_0}^{t_0+T_0+t_r}e^{\lambda_a (t_0+T_0+t_r-\tau)}d\tau$.
Overall, from $t_0$ to $t_0+T_0+t_r$, we have 
\begin{align}
   & V(t_0+T_0+t_r) \le
     V(t_0)e^{-\lambda t_a+\lambda_a(T_0-t_a+t_r)} + C
\end{align} where $C =\overline{w}_0e^{-\lambda t_a+\lambda_a (T_0-t_a-t_c+t_r)} + \overline{w}_1e^{\lambda_a (T_0-t_a-t_c+t_r)}$ $ + (\overline{u} + \overline{w}_2)e^{\lambda_a t_r} + \overline{w}_3$.
Let $p_a$ be the probability density function of $t_a$ whose support is $[0,T_0-t_c]$, $V$ satisfies
\begin{align}
\nonumber
    &\mathbb{E}[V(t_0+T_0+t_r)] \le 
    V(t_0)\mathbb{E}_{t_a}[e^{-(\lambda+\lambda_a) t_a+\lambda_a(T_0+t_r)}] + g(t_r)
\end{align} where $\mathbb{E}_{t_a}[e^{-(\lambda +\lambda_a) t_a+\lambda_a(T_0+t_r)}] =  e^{\lambda_a(T_0+t_r)}\int_0^{T_0-t_c}$ $   e^{-(\lambda +\lambda_a) t_a} p_a(t_a) dt_a$ and $g(t_r) = \int_0^{T_0-t_c}  p_a(t_a) C dt_a$.
The overall temporal evolution of the CPS can be modeled as a composite of many such events. To ensure the right-hand side of the above inequality to be finite, we need \eqref{eq:reboot}.  Therefore, \eqref{eq:reboot} is a sufficient condition for the mean-square boundedness of states in \eqref{cps}.
\end{proof}

\begin{remark}
    \label{rem:basic}
    When $t_a\equiv T_0-t_c$ (i.e., in the absence of an attack), \eqref{eq:reboot} simplifies  $T_0 > t_c+\frac{\lambda_a}{\lambda}(t_c+t_r)$. This implies that the CPS states can remain mean-square bounded if a single controller operates for a minimum duration of $t_c+\frac{\lambda_a}{\lambda}(t_c+t_r)$  before each re-initialization. Let $\underline t$ denote the minimum duration required for a successful attack, \eqref{eq:reboot} is satisfied if $\underline t > \frac{\lambda_a}{\lambda+\lambda_a}(T_0+t_r)$. Additionally,
    let $\overline{t}$ with $ \overline{ t}<T_0-t_c$ denote the maximum duration required for a successful attack, \eqref{eq:reboot} will not hold if $\overline{t} \le \frac{\lambda_a}{\lambda+\lambda_a}(T_0+t_r)$.   In summary, the defender can make it more challenging for the attacker to succeed by minimizing $ \frac{\lambda_a}{\lambda+\lambda_a}(T_0+t_r)$, which can be achieved by reducing the controller's operational duration $T_0$.
\end{remark}

\begin{remark}
    \label{rem:impact}
    According to the proof of {\bf Theorem~\ref{theorem2}}, we can evaluate the impact of an attack on the CPS  in the absence of any resilient control strategies,  under {\bf Assumption~\ref{assump4}}. Let $\overline{ t}$ continue to represent the maximum duration required for a successful attack, at time $t>\overline{t}$, the temporal expectation of the Lyapunov function with compromised control signal satisfies the following:
    \begin{align}
        \mathbb{E}_{t_a^\prime}[V(t)] \le C_1e^{\lambda_a t}+C_2
        \label{eq:exponentially}
    \end{align} where $C_1$ and $C_2$ are some constants. \eqref{eq:exponentially} implies that the attacker could destabilize the CPS states so that the Lyapunov function increases exponentially.  Therefore, {\bf Assumption~\ref{assump4}} is mild and physically reasonable. 
\end{remark}

\subsection{Performance of the Defense in Fig.~\ref{fig:rcp}}
 
\begin{thm}
\label{theorem1}
Under the {\bf Assumptions \ref{assumption1}$-$\ref{assumption3}} and scenario in Fig.~\ref{fig:rcp},  a sufficient condition for the CPS states in \eqref{cps} to be mean-square bounded is 
\begin{align}
    \mathbb{E}_{t_a,t_d}[e^{-\lambda t_a+\lambda_a(t_d+t_r+t_c)}] < 1.
    \label{condition}
\end{align}
where $t_a=\min\{t_a^\prime, T_0-t_c\}$, $t_d=\min\{t_d^\prime, T_0-t_a-t_c\}$,  $\lambda = \underline{\alpha} - \frac{1}{\epsilon}\overline{\beta}_1\overline{\alpha}$ with $\epsilon$ being any positive constant that satisfies $\epsilon > \frac{1}{\underline{\alpha}}\overline{\beta}_1\overline{\alpha}$, and $\lambda_a = \overline{\gamma}_1 + \frac{1}{\epsilon_a}\overline{\gamma}_2+\frac{1}{\epsilon_b}\overline{\beta}_2$ with  $\epsilon_a$ and $\epsilon_b$ being any positive constants. 
\end{thm}
\begin{proof}
The proof method is similar to the proof of {\bf Theorem}~\ref{theorem2}.
\end{proof}

\begin{remark}
    \label{rem:worst}
    Consider the worst-case scenario where the attack evades the anomaly detector, then \eqref{condition} simplifies to \eqref{eq:reboot} since $t_d \equiv T_0-t_a-t_c$. Apart from the worst-case scenario, the inclusion of an anomaly detector refines the re-initialization defense illustrated in Fig.~\ref{fig:a-d-model}. On the other hand, consider the constant-time detection with $t_d\equiv \mu_d$ and constant-time attack with $t_a \equiv \mu_a$, then  \eqref{condition} simplifies to $\mu_d < \frac{\lambda}{\lambda_a}\mu_a - t_r - t_c$.  In this context, a detection time  shorter than $\frac{\lambda}{\lambda_a}\mu_a - t_r - t_c$ guarantees that the CPS states are mean-square bounded. To achieve this, the defender can reduce $\lambda_a$, $t_r$, $t_c$, or $\mu_d$.
\end{remark}

\subsection{Performance of  Alg.~\ref{alg:RRS2}}

To analyze the defense depicted in Fig.~\ref{fig:RRS2}, we consider the worst-case scenario where the attack successfully bypasses the anomaly detector. In such a case, Alg.~\ref{alg:RRS2} essentially reverts to the defense outlined in Fig.~\ref{fig:alg2}. We aim to demonstrate that even this pared-down defense strategy can prove effective. Furthermore, while it's conceivable for different controllers to employ various control laws, we focus on the worst-case scenario where all controllers operate using the same control law, specifically the one with the slowest rate of convergence.

\begin{thm}
\label{theoremrrs}
Under the {\bf Assumptions~\ref{assumption1}$- $\ref{assumption4}}, a sufficient condition for the CPS states in \eqref{cps} to be mean-square bounded given $n$ controllers and Alg.~\ref{alg:RRS2} is 
\begin{align}
    \mathbb{E}_{t_a}[e^{-(\lambda+\lambda_a) {t_a}+\lambda_a \frac{t_r}{n-1}}] < 1
    \label{eq:basicRRS}
\end{align} where $t_a = \min\{t_a^\prime, \frac{t_r}{n-1}-t_c\}$, $\lambda = \underline{\alpha} - \frac{1}{\epsilon}\overline{\beta}_1\overline{\alpha}$ with $\epsilon$ being any positive constant that satisfies $\epsilon > \frac{1}{\underline{\alpha}}\overline{\beta}_1\overline{\alpha}$, and $\lambda_a = \overline{\gamma}_1 + \frac{1}{\epsilon_a}\overline{\gamma}_2+\frac{1}{\epsilon_b}\overline{\beta}_2$ with $\epsilon_a$ and $\epsilon_b$ being any positive constants.
\end{thm}
\begin{proof}
\label{proofrrs}
 Alg.~\ref{alg:RRS2} ensures that the CPS will not use controllers that are in re-initialization. Therefore, from $[t_0, t_0+t_c]$ the CPS follows \eqref{assump4} with $u\equiv0$. From $[t_0+t_c,t_0+t_c+t_a]$, the CPS follows  (\ref{assump3}) with $u=u_c(x)$. During $[t_0+t_c+t_a, t_0+\frac{t_r}{n-1}]$, the CPS follows (\ref{assump4}) with changed control signal $u=u_{max}$.   Using a similar method in {\bf Theorem~\ref{theorem2}}, we get a similar condition for the mean-square boundedness of the CPS states, i.e.,
\begin{align}
  \mathbb{E}_{t_a}[e^{-(\lambda+\lambda_a) t_a + \lambda_a \frac{t_r}{n-1} }] \le 1. 
    \label{eq:rrs}
\end{align}
\end{proof}
\newcommand{\floor}[1]{\left\lfloor #1 \right\rfloor}

\begin{remark}
    \label{rm:original}
    In the case where the stealthy attack is instantaneous (i.e., $t_a=0$), affects fewer than $n$ controllers (i.e., $m<n$ controllers are compromised), and cannot be removed by re-initialization, the condition \eqref{eq:basicRRS} simplifies to $m < \frac{\lambda}{\lambda+\lambda_a}n$. Under this condition, the maximum number of compromised controllers is limited to $\floor{\frac{\lambda}{\lambda+\lambda_a}n}$. If $\lambda = \lambda_a$, then no more than half of the total controllers can be compromised. This finding aligns with previous results  in \cite{KK20,KK21}. The upper bound for the number $n$ of involved controllers is influenced by the authentication time $t_c$ and re-initialization period $t_r$. Specifically, since we require $t_a > 0$, it implies $\frac{t_r}{n-1}>t_c$ and $n < 1+\frac{t_r}{t_c}$. If $n$ exceeds this threshold, the CPS will be preoccupied with the authentication process for various controllers, failing to actually engage them in controlling the system.
\end{remark}

{\bf Theorems~\ref{theorem2}$-$\ref{theoremrrs}} are sufficient conditions, meaning that even if these conditions are not met, the CPS states could still potentially remain mean-square bounded. \textcolor{\colorname}{ This work only considers sufficient conditions due to the high complexity of the overall system. Specifically, the considered CPS includes stochastic and discrete-event components in addition to the underlying continuous-time system dynamics. Hence, while necessary and sufficient conditions for stability would be nice to have, obtaining such conditions is highly difficult/infeasible. Indeed, even for simpler scenarios such as design of stabilizing controllers for classes of nonlinear systems, it is often the case that only sufficient conditions can be obtained for conditions for such control designs.  Nevertheless, from a practical standpoint, the obtained sufficient conditions are valuable since they indicate conditions under which the proposed approach is applicable. Furthermore, given that our paper focuses on developing a redundancy-based scheme for robustness of CPS to attacks, sufficient conditions are most important since they indicate (albeit possibly conservative) conditions under which the proposed approach provides the robustness benefits. 
}

\subsection{\textcolor{\colorname}{Unbounded Re-Initialization Period}
}
\textcolor{\colorname}{
When the re-initialization period $t_r$ or the authentication time $t_c$ are bounded random variables, {\bf Theorems~\ref{theorem2}$-$\ref{theoremrrs}} still apply by using the maximum restarting or authenticating time.  The following theorem derives the the number of normal controllers required for the CPS to be mean-square bounded when the re-initialization period $t_r$ is unbounded. }

\textcolor{\colorname}{
\begin{thm}
\label{thm:requirement}
Under the {\bf Assumptions~\ref{assumption1}$- $\ref{assumption4}}, given Alg.~\ref{alg:RRS2} and $n$ controllers with $n_1$ of them having finite re-initialization time $t_r$  and $n-n_1$ of them having infinite re-initialization time (i.e., $t_r=\infty$), a sufficient condition for the CPS states to be mean-square bounded is 
\begin{align}
\label{thm:minimum}
    n_1 > \frac{ \lambda_a\frac{nt_r}{n-1}}{\frac{t_r}{n-1} \lambda_a - \text{ln } \mathbb{E}_{t_a}[e^{-(\lambda+\lambda_a) {t_a}+\lambda_a \frac{t_r}{n-1}}]}
\end{align} where $\lambda$ and $\lambda_a$ have the same meaning as in {\bf Theorem~\ref{theoremrrs}}. Moreover, when $t_a=\frac{t_r}{n-1}$, we have
\begin{align}
\label{thm:simple}
    n_1 > \frac{\lambda_a}{\lambda+\lambda_a}n.
\end{align}
\end{thm}
\begin{proof}
Alg.~\ref{alg:RRS2} ensures that each controller is used only for $T_0$ time and then will be re-initialized. For controllers with infinite re-initialization time, the control signal will be $0$ once it is reused by Alg.~\ref{alg:RRS2}. Therefore, to ensure the Lyapunov function to be mean-square bounded after iterating all $n$ controllers, we must have
  \begin{align}
        (\mathbb{E}_{t_a}[e^{-(\lambda+\lambda_a) {t_a}+\lambda_a \frac{t_r}{n-1}}])^{n_1} (e^{\lambda_a \frac{t_r}{n-1}})^{n-n_1}<1.
  \end{align} This implies \eqref{thm:minimum}. Furthermore, when $t_a=\frac{t_r}{n-1}$ is a constant, \eqref{thm:minimum} becomes \eqref{thm:simple}.
\end{proof}
\begin{remark}
    {\bf Theorem~\ref{thm:requirement}} implies that the number of controllers with infinite re-initialization time must be limited to ensure the mean-square boundedness of the CPS states.  If $\lambda = \lambda_a$, then no more than half of the total controllers can have infinite re-initialization time.   {\bf Theorem~\ref{thm:requirement}} also applies for the scenario where $n_1$ out of $n$ controllers have finite authentication time $t_c$ and $n-n_1$ of them have infinite authentication time. 
\end{remark}
}

\section{Simulation Studies}
\label{sec:exp}

This section provides empirical illustrations for {\bf Theorems~\ref{theorem2}$-$\ref{theoremrrs}} and Alg.~\ref{alg:RRS2}, using two systems cited from \cite{KK21}.  Each experiment is conducted ten times, and the average results are plotted.

\subsection{Example of A Third-Order System}

Consider the  system \cite{KK21}
\begin{align}
\nonumber
    \dot x_1  &= x_2; \qquad
    \dot x_2  = -x_2 + \text{sin}(0.1x_1)w + x_3;  \\
    \dot x_3  &= -x_3 + u
    \label{sys1}
\end{align} where $x=[x_1, x_2, x_3]^T$ is the system state with initial value $x_0 = [5,2,2]^T$, $u$ is the control input, and $w=\text{sin}(0.2t)$ is the disturbance. Define the controller as $u = Kx$, where $K = [-27, -19, -7]$ with the three desired pole locations at $-3$. We set the re-initialization time $t_r=1$s, the authentication time $t_c=0.01$s, and the physical constraint $u_{max} = 10$. A controller outputs zero when it is in re-initialization or authentication. The attacked controllers will always output the maximum value (i.e., $u\equiv 10$).  We assume that the attack and detection distributions are both truncated Gaussians and use four parameters $(a,b,\mu,\sigma)$ to depict them, i.e., $p(t)\sim \frac{1}{\sigma} \frac{\phi(\frac{t-\mu}{\sigma})}{\Phi(\frac{b-\mu}{\sigma}) - \Phi(\frac{a-\mu}{\sigma})}$, where $\phi(x) = \frac{1}{\sqrt{2\pi}} exp(-\frac{1}{2}x^2)$ and $\Phi(x) = \frac{1}{2}\left(1+erf(\frac{x}{\sqrt{2}})\right)$\footnote{$exp(x)=e^x$ and $erf(x)=\frac{2}{\sqrt{\pi}}\int_0^xe^{-t^2}dt$.}.  The re-initialization process is to simply turn off and on the controller since it is static. We allow the attacker to target any subset of controllers simultaneously.

\textbf{Calculating the Parameters:}
Assume the Lyapunov function is $V = x^T P x$.  Let $B = [0,0,1]^T$, $f(x,w) = [0, \text{sin}(0.1x_1)w, 0]^T$, $A = [[0,0,0]^T, [1,-1,0]^T, [0,1,-1]^T]$ and $A_c = [[0,0,-27]^T$, $[1,-1,-19]^T, [0,1,-8]^T]$.
Then we can write  system \eqref{sys1} as $\dot x = Ax+Bu+f$ and get $\dot V  = x^T(A^TP+PA)x+2x^TPBu+2x^TPf$.
If we apply the control law $u = Kx$, then system \eqref{sys1} becomes $\dot x = A_cx + f$ and we have $\dot V =  x^T(A_c^TP+PA_c)x + 2x^TPf$.
Let $I_3$ be the $3\times 3$ identity matrix. Solving the equation $A_c^TP+PA_c = -I_3$ to get $P$ so that  $ \dot V  = -x^T I_3 x + 2x^TPf$.
We can let $\alpha(|x|) = x^Tx$, $\beta_1(|x|)=2\sqrt{x^Tx}$, and $\mu_1(|w|) = |w|$. Similarly, we can let $\gamma_1(|x|)=4x^Tx$, $\gamma_2(|x|)=0.5\sqrt{x^Tx}$, $\beta_2(|x|)=2\sqrt{x^Tx}$, $\gamma_u(|u|)=|u|$, and $\mu_2(|w|)=|w|$. Then we have $\underline{\alpha} = 0.2$, $\overline{\alpha} = 0.25$, $\overline{\beta}_1=4$, $\overline{\gamma}_1 = 1$, $\overline{\gamma}_2=0.25$, and $\overline{\beta}_2 = 1$. We choose $\epsilon = 10 \overline{\beta}_1 \overline{\alpha} / \underline{\alpha}  $ so that $\lambda = 0.9\underline{\alpha} = 0.18 $ and choose $\epsilon_a=\epsilon_b=10$ so that $\lambda_a = \overline{\gamma}_1 + \overline{\gamma}_2/\epsilon_a + \overline{\beta}_2/\epsilon_b = 1.125$. 

{\bf No Resilient Control Strategy}:  The attack distribution is set as a truncated Gaussian with parameters $(0, 1, 0.1, 0.1)$. The  attack's impact without any resilient control measures is shown in Fig.\ref{fig:noresilient}. The result shows that  the attacker is  capable of diverting the CPS from its intended behavior under {\bf Assumptions~\ref{assumption2}} and {\bf \ref{assumption4}}.

\begin{figure}[ht]
    \centering
    \includegraphics[width=\linewidth]{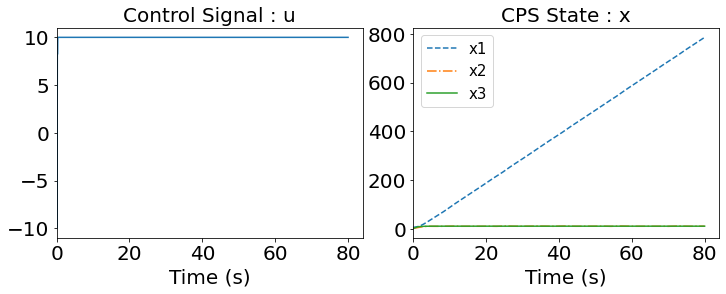}
    \caption{Attack's impact without any resilient control strategy.}
    \label{fig:noresilient}
\end{figure}

{\bf Alg.~\ref{alg:RRS2} without Anomaly Detection}: 
We initially consider the case with a single controller (i.e., the defense depicted in Fig.\ref{fig:a-d-model}). We set $T_0=1$s, yielding $\mathbb{E}_{t_a}[e^{-(\lambda+\lambda_a){t_a}+\lambda_a(T_0+t_r)}]\approx 5.5$. The resulting CPS states are not mean-square bounded, as illustrated in the first row of Fig.\ref{fig:basiclinear}. Nevertheless, when compared to Fig.~\ref{fig:noresilient}, this defense strategy causes the CPS states to deviate less from their intended values over the same time frame. When applying Alg.\ref{alg:RRS2} with $n=4$, the term $\mathbb{E}_{t_a}[e^{-(\lambda+\lambda_a){t_a}+\lambda_a \frac{t_r}{n-1}}]$ is approximately 1.03. These results are presented in the second row of Fig.\ref{fig:basiclinear}. Relative to the single-controller case ($n=1$), the impact of the attack is further mitigated. As $n$ increases to 11, the term becomes approximately 1. The CPS states are mean-square bounded. 

\begin{figure}[ht]
    \centering
    \includegraphics[width=\linewidth]{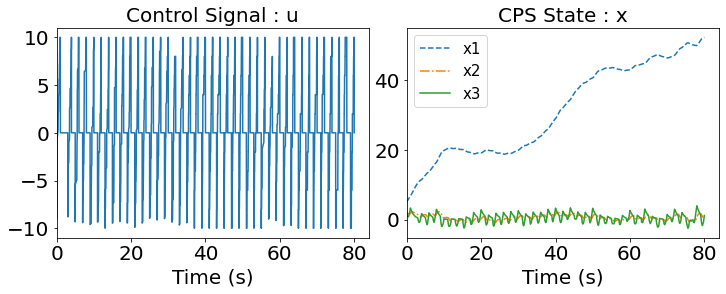}
    \includegraphics[width=\linewidth]{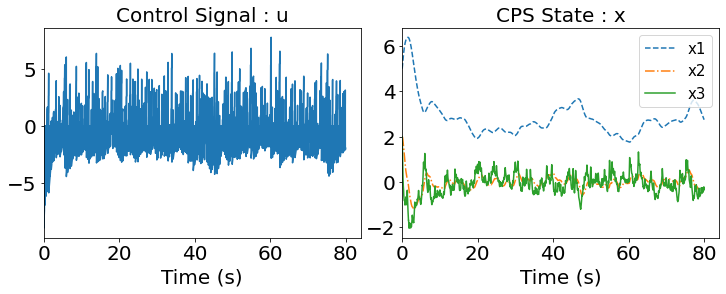}
    \includegraphics[width=\linewidth]{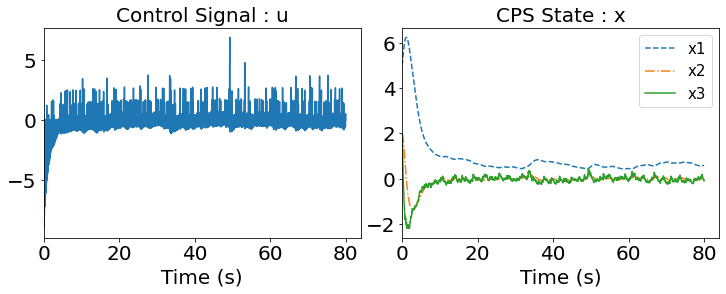}
    \caption{Performance of Alg.~\ref{alg:RRS2} in the absence of anomaly detection: The compromised control signals are set to $u=u_{max}=10$. All experiments are conducted ten times, and the average results are plotted. Therefore, the compromised control signal may show an average value of less than ten. The first row: $n=1$ with $T_0=1$. The middle row: $n=4$. The last row: $n=11$.}
    \label{fig:basiclinear}
\end{figure}

{\bf Alg.~\ref{alg:RRS2} with Anomaly Detection}: 
The detection distribution is a truncated Gaussian with parameters $(0, 1, 0.1, 1)$, resulting in $\mathbb{E}_{t_a,t_d}[e^{-\lambda t_a+\lambda_a(t_d+t_r+t_c)}]$ $ \approx 4.3$ for a single-controller case. The outcome is displayed in the first row of Fig.\ref{fig:anomalylinear}. Compared to the first row in Fig.\ref{fig:basiclinear}, the inclusion of the anomaly detector mitigates the attack's impact.  When $n=4$ controllers are utilized, the anomaly detector mitigates the attack's impact during each controller's working period, as can be seen by comparing  the middle rows of Fig.\ref{fig:anomalylinear} and Fig.\ref{fig:basiclinear}. When comparing the last rows of Fig.\ref{fig:anomalylinear} and Fig.~\ref{fig:basiclinear}, the advantage of employing the anomaly detector diminishes as the number of controllers ($n$) increases to a sufficient level (e.g., 11).  However, this shows that using sufficient controllers can mitigate the impact of stealthy attacks that bypass the anomaly detector.

\begin{figure}[ht]
    \centering
    \includegraphics[width=\linewidth]{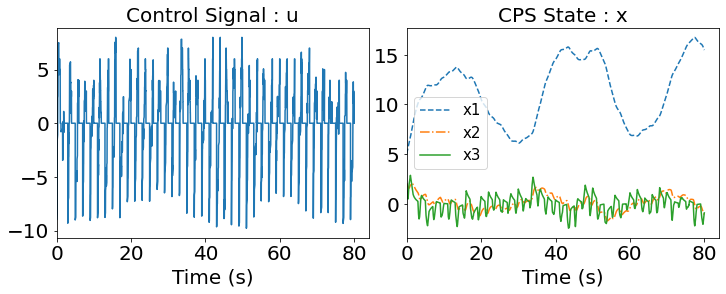}
    \includegraphics[width=\linewidth]{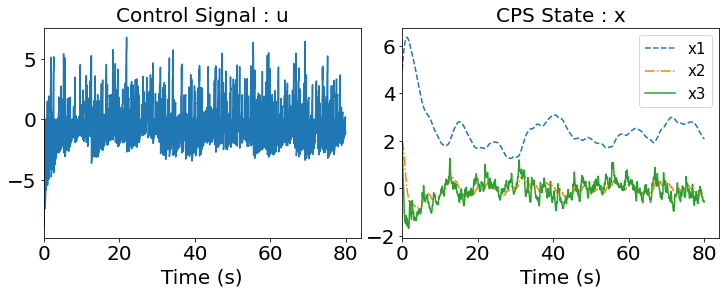}
    \includegraphics[width=\linewidth]{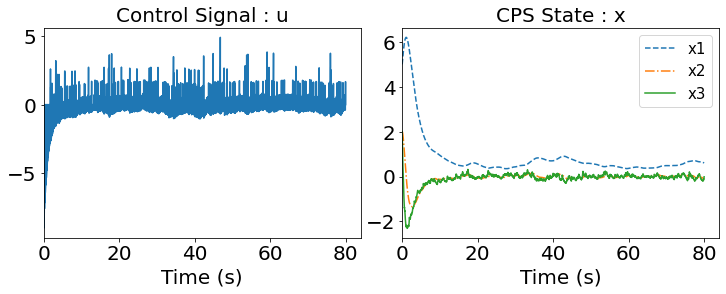}
    \caption{Performance of Alg.~\ref{alg:RRS2} with anomaly detection. For compromised control signal, $u=u_{max}=10$. We run all the experiments ten times and plot the average result. First row: $n=1$ and $T_0=1$. Middle row: $n=4$. Last row: $n=11$.}
    \label{fig:anomalylinear}
\end{figure}

\subsection{Example of A SMIB System}

Consider the SMIB power system \cite{WHMG93,JKF94} shown in Fig.~\ref{fig:bus} with the following dynamics:
\begin{figure}[ht]
    \centering
    \includegraphics[width=\linewidth]{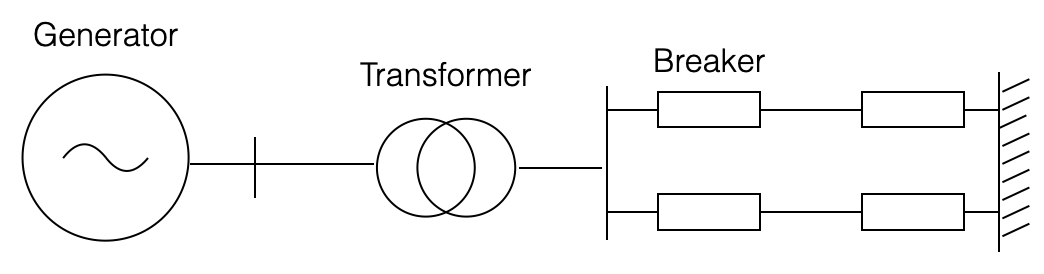}
    \caption{A SMIB power system: generator connected through transmission lines to an infinite bus.}
    \label{fig:bus}
\end{figure}
\begin{align}
\nonumber
    \dot \delta & = \omega; \qquad
    \dot \omega = \frac{1}{2H}[-D\omega + \omega_0(P_{m0} - P_e)]; \\
    \dot E_{q1} & = \frac{1}{T_{d0}}[E_f - E_q - \hat E_q]
    \label{SMIB}
\end{align} 
where $\delta$, $\omega$, and $E_{q1}$ are the system states representing the power angle, relative speed, and quadrature-axis transient electromotive force (EMF), respectively. $E_q$ is the quadrature-axis EMF, $\hat E_q$ is the disturbance, and $E_f$ is the equivalent EMF in the excitation coil, which is the control signal. $P_{m0}$ represents the mechanical input power, $\omega_0 = 2\pi f_0$ represents the synchronous machine speed, and $P_e$ represents the active electrical power delivered by the generator. $T_{d0}$ is the direct-axis transient short circuit time constant. The parameters $D$ and $H$ are the per unit (p.u.) damping constant and inertia constant, respectively. $P_e$, $E_q$, and $E_{q1}$ are related through the algebraic equations \cite{WHMG93} $ E_q  = \frac{x_{ds}}{x_{ds1}}E_{q1} - \frac{x_d - x_{d1}}{x_{ds1}}V_s \text{cos}\delta $ and $P_e = \frac{V_sE_q\text{sin}\delta}{x_{ds}}$ where $V_s$ is the infinite bus voltage, and $x_d$, $x_{d1}$, $x_{ds}$, and $x_{ds1}$ are reactance parameters as in \cite{WHMG93}. We design the control law $u=E_f$ based
on external feedback linearization, i.e., $v_f = K[\delta - \delta_0, \omega, P_e - P_{m0}]^T + P_{m0}$ and $E_f = \frac{1}{I_q} (v_f - \frac{x_d-x_{d1}}{x_{ds1}}T_{d01}I_qV_s\omega\text{sin}\delta  - \frac{V_sT_{d01}}{x_{ds}}E_q\omega\text{cos}\delta)$
where $I_q = \frac{V_s \text{sin}\delta}{x_{ds}}$ and $T_{d01} = \frac{x_{ds1}}{x_{ds}}T_{d0}$. $K$ is the $1\times 3$ gain vector and $\delta_0$ is the desired operating point for the power angle $\delta$. We set $P_{m0} = 0.9 $ p.u., $\omega_0 = 314.159$ rad/s, $T_{d0} = 6.9$ s, $D = 5$ p.u., $H =4 $ s, $V_s= 1.0$ p.u., $x_d=1.863$, $x_{d1}=0.257$, $x_{ds}=2.2327$, $x_{ds1}=0.6267$,  $K = [19.3, 6.43, -47.6]^T$, $u_{max}=2.3$ p.u., and $\delta_0 = 1.309$ rad. We allow the attacker to target any subset of controllers simultaneously. The  compromised signal is $u_{max}$.  The  initial  condition  is $[\delta, \omega, E_{q1}]^T = [1, 1, 1]^T$. We set  $\hat E_q = 0.01 \text{cos} \frac{\pi t}{2}$, $t_r = 0.35$s, and  $t_c=0.05$s. 

\begin{figure}[ht]
    \centering
    \includegraphics[width=\linewidth]{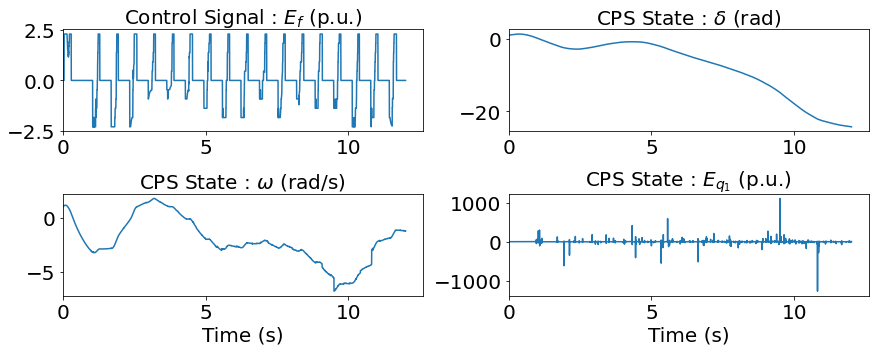}
    \includegraphics[width=\linewidth]{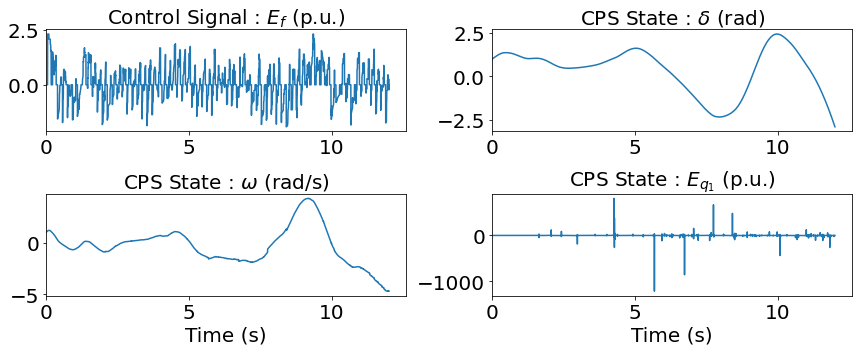}
    \includegraphics[width=\linewidth]{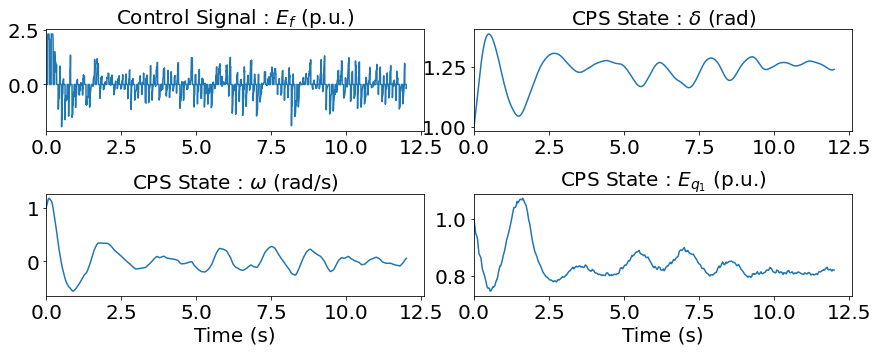}
    \includegraphics[width=\linewidth]{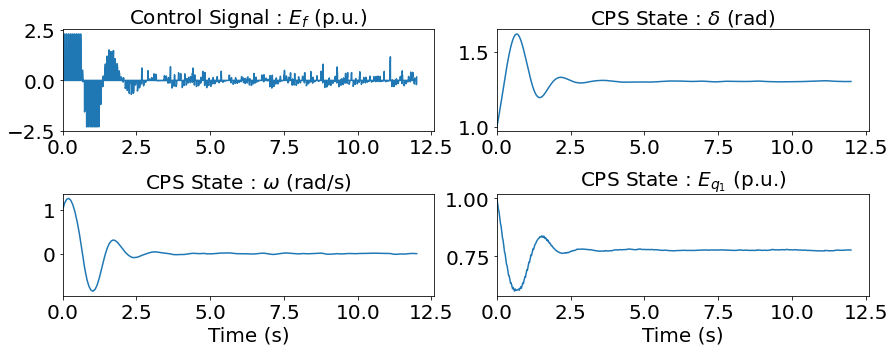}
    \caption{The performance of Alg.~\ref{alg:RRS2}. The attacker alters the control signal to $u=u_{max}=2.3 $ $p.u.$. All experiments are conducted ten times, and the average results are plotted. Therefore, the compromised control signal may show an average value of less than $u_{max}$. First row: $n=1$ and $T_0=0.3$s. Second row: $n=3$. Third row: $n=4$. Last row: $n=6$.}
    \label{fig:nonlinearcase}
\end{figure}

We set both the attack and detection distributions as $(0, 0.2, 0.15, 1)$.  The performance of the CPS under Alg.\ref{alg:RRS2} with different numbers of controllers is presented in Fig.\ref{fig:nonlinearcase}. Our Alg.\ref{alg:RRS2} not only mean-square bounds the CPS states but also achieves mean-square stabilization (i.e., $\lim_{t\to\infty}\mathbb{E}[V(t)] = 0$) when $n$ increases to 6, which is a stronger condition than the mean-square boundedness. The explanation is that as the number of utilized controllers increases, the period during which the controller is exposed to attacks decreases. The working controller is likely to be replaced before being compromised. This example demonstrates the efficacy and potential of Alg.\ref{alg:RRS2} in practical applications.

\section{Conclusion}
\label{sec:conclusion}
A multi-controller switching strategy integrated with anomaly detection to identify attacks and periodic re-initialization to remove attacks is proposed for enhancing the CPS resiliency.  An integrated attack-defense model is proposed for the performance analysis of our approach. Three sufficient conditions for   the mean-square boundedness of the CPS states are derived. Our switching strategy is  validated on two different systems. Future research directions include broadening the types of dynamic systems that can be addressed and relaxing the assumptions required for the current model.

\bibliographystyle{plain}
\bibliography{ref}

\begin{thebibliography}{10}

\bibitem{ACHLMC18}
Fardin Abdi, Chien-Ying Chen, Monowar Hasan, Songran Liu, Sibin Mohan, and Marco Caccamo.
\newblock Guaranteed physical security with restart-based design for cyber-physical systems.
\newblock In {\em 2018 ACM/IEEE 9th International Conference on Cyber-Physical Systems}, pages 10--21, 2018.

\bibitem{A21}
Qasem~Abu Al-Haija.
\newblock On the security of cyber-physical systems against stochastic cyber-attacks models.
\newblock In {\em Proceedings of the IEEE International IOT, Electronics and Mechatronics Conference}, pages 1--6, Toronto, Canada, 2021.

\bibitem{ANR16}
Cesare Alippi, Stavros Ntalampiras, and Manuel Roveri.
\newblock Model-free fault detection and isolation in large-scale cyber-physical systems.
\newblock {\em IEEE Transactions on Emerging Topics in Computational Intelligence}, 1(1):61--71, 2016.

\bibitem{AESE21}
Abdulwahab Almutairi, H~El-Metwally, MA~Sohaly, and IM~Elbaz.
\newblock Lyapunov stability analysis for nonlinear delay systems under random effects and stochastic perturbations with applications in finance and ecology.
\newblock {\em Advances in Difference Equations}, 2021(1):1--32, 2021.

\bibitem{AZKYS19}
Miguel~A Arroyo, M~Tarek~Ibn Ziad, Hidenori Kobayashi, Junfeng Yang, and Simha Sethumadhavan.
\newblock Yolo: frequently resetting cyber-physical systems for security.
\newblock In {\em Autonomous Systems: Sensors, Processing, and Security for Vehicles and Infrastructure 2019}, volume 11009, pages 166--183, 2019.

\bibitem{B43}
Richard Bellman.
\newblock The stability of solutions of linear differential equations.
\newblock {\em Duke Mathematical Journal}, 10(4):643--647, 1943.

\bibitem{GSYGA04}
George Candea, Shinichi Kawamoto, Yuichi Fujiki, Greg Friedman, and Armando Fox.
\newblock {Microreboot{\textemdash}A} technique for cheap recovery.
\newblock In {\em Proceedings of Symposium on Operating Systems Design \& Implementation}, 2004.

\bibitem{CALHHS11}
Alvaro~A C{\'a}rdenas, Saurabh Amin, Zong-Syun Lin, Yu-Lun Huang, Chi-Yen Huang, and Shankar Sastry.
\newblock Attacks against process control systems: risk assessment, detection, and response.
\newblock In {\em Proceedings of the 6th ACM Symposium on Information, Computer and Communications Security}, pages 355--366, 2011.

\bibitem{ENK11}
David Evans, Anh Nguyen-Tuong, and John Knight.
\newblock Effectiveness of moving target defenses.
\newblock {\em Moving Target Defense: Creating Asymmetric Uncertainty for Cyber Threats}, pages 29--48, 2011.

\bibitem{GUCVFRTSC18}
Jairo Giraldo, David Urbina, Alvaro Cardenas, Junia Valente, Mustafa Faisal, Justin Ruths, Nils~Ole Tippenhauer, Henrik Sandberg, and Richard Candell.
\newblock A survey of physics-based attack detection in cyber-physical systems.
\newblock {\em ACM Computing Surveys}, 51(4):1--36, 2018.

\bibitem{G19}
Thomas~Hakon Gronwall.
\newblock Note on the derivatives with respect to a parameter of the solutions of a system of differential equations.
\newblock {\em Annals of Mathematics}, 20(4):292--296, 1919.

\bibitem{HO18}
Farshad Harirchi and Necmiye Ozay.
\newblock Guaranteed model-based fault detection in cyber--physical systems: A model invalidation approach.
\newblock {\em Automatica}, 93:476--488, 2018.

\bibitem{HMS03}
Desmond~J Higham, Xuerong Mao, and Andrew~M Stuart.
\newblock Exponential mean-square stability of numerical solutions to stochastic differential equations.
\newblock {\em LMS Journal of Computation and Mathematics}, 6:297--313, 2003.

\bibitem{JKF94}
Sandeep Jain, Farshad Khorrami, and B~Fardanesh.
\newblock Adaptive nonlinear excitation control of power systems with unknown interconnections.
\newblock {\em IEEE Transactions on Control Systems Technology}, 2(4):436--446, 1994.

\bibitem{KSCKMK16}
Anastasis Keliris, Hossein Salehghaffari, Brian Cairl, Prashanth Krishnamurthy, Michail Maniatakos, and Farshad Khorrami.
\newblock Machine learning-based defense against process-aware attacks on industrial control systems.
\newblock In {\em Proceedings of IEEE International Test Conference}, pages 1--10, 2016.

\bibitem{KKK16}
Farshad Khorrami, Prashanth Krishnamurthy, and Ramesh Karri.
\newblock Cybersecurity for control systems: a process-aware perspective.
\newblock {\em IEEE Design \& Test}, 33(5):75--83, 2016.

\bibitem{KM01}
Ilya Kolmanovsky and Tatiana~L Maizenberg.
\newblock Mean-square stability of nonlinear systems with time-varying, random delay.
\newblock {\em Stochastic Analysis and Applications}, 19(2):279--293, 2001.

\bibitem{KXWSL18}
Fanxin Kong, Meng Xu, James Weimer, Oleg Sokolsky, and Insup Lee.
\newblock Cyber-physical system checkpointing and recovery.
\newblock In {\em Proceedings of International Conference on Cyber-Physical Systems}, pages 22--31, 2018.

\bibitem{KK20}
Prashanth Krishnamurthy and Farshad Khorrami.
\newblock Adaptive randomized controller switching for resilient cyber-physical systems.
\newblock In {\em Proceedings of the IEEE Conference on Control Technology and Applications}, pages 738--743, Montreal, Canada, 2020.

\bibitem{KK21}
Prashanth Krishnamurthy and Farshad Khorrami.
\newblock Resilient redundancy-based control of cyber-physical systems through adaptive randomized switching.
\newblock {\em Systems \& Control Letters}, 158:105066, 2021.

\bibitem{KKKPS18}
Prashanth Krishnamurthy, Farshad Khorrami, Ramesh Karri, David Paul-Pena, and Hossein Salehghaffari.
\newblock Process-aware covert channels using physical instrumentation in cyber-physical systems.
\newblock {\em IEEE Transactions on Information Forensics and Security}, 13(11):2761--2771, 2018.

\bibitem{LHBF14}
Per Larsen, Andrei Homescu, Stefan Brunthaler, and Michael Franz.
\newblock Sok: Automated software diversity.
\newblock In {\em Proceedings of IEEE Symposium on Security and Privacy}, pages 276--291, 2014.

\bibitem{LHKK18}
Nandi~O Leslie, Richard~E Harang, Lawrence~P Knachel, and Alexander Kott.
\newblock Statistical models for the number of successful cyber intrusions.
\newblock {\em The Journal of Defense Modeling and Simulation}, 15(1):49--63, 2018.

\bibitem{MCMK19}
J~Sukarno Mertoguno, Ryan~M Craven, Matthew~S Mickelson, and David~P Koller.
\newblock A physics-based strategy for cyber resilience of cps.
\newblock In {\em Autonomous Systems: Sensors, Processing, and Security for Vehicles and Infrastructure}, volume 11009, pages 79--90, 2019.

\bibitem{MCS13}
Yilin Mo, Rohan Chabukswar, and Bruno Sinopoli.
\newblock Detecting integrity attacks on scada systems.
\newblock {\em IEEE Transactions on Control Systems Technology}, 22(4):1396--1407, 2013.

\bibitem{PMS10}
Sujit~S Phatak, DJ~McCune, and George Saikalis.
\newblock Cyber physical system: A virtual cpu-based mechatronic simulation.
\newblock {\em IFAC Proceedings Volumes}, 43(18):405--410, 2010.

\bibitem{RLNW16}
James~F Riordan, Richard~P Lippmann, Sebastian~J Neumayer, and Neal Wagner.
\newblock A model of network porosity.
\newblock Technical report, MIT Lincoln Laboratory Lexington United States, 2016.

\bibitem{SSTC09}
Zhen Song, Chellury~Ram Sastry, Nazif~Cihan Tas, and YangQuan Chen.
\newblock Feasibility analysis on optimal sensor selection in cyber-physical systems.
\newblock In {\em Proceedings of the American Control Conference}, pages 5368--5373, St. Louis, MO, 2009.

\bibitem{TS12}
Angel Tocino and MJ~Senosiain.
\newblock Mean-square stability analysis of numerical schemes for stochastic differential systems.
\newblock {\em Journal of Computational and Applied Mathematics}, 236(10):2660--2672, 2012.

\bibitem{WHMG93}
Youyi Wang, David~J Hill, Richard~H Middleton, and Long Gao.
\newblock Transient stability enhancement and voltage regulation of power systems.
\newblock {\em IEEE Transactions on Power Systems}, 8(2):620--627, 1993.

\bibitem{WK12}
Jin Wei and Deepa Kundur.
\newblock A flocking-based model for dos-resilient communication routing in smart grid.
\newblock In {\em Proceedings of the IEEE Global Communications Conference}, pages 3519--3524, Anaheim, CA, 2012.

\bibitem{WYPSLW21}
Chengwei Wu, Weiran Yao, Wei Pan, Guanghui Sun, Jianxing Liu, and Ligang Wu.
\newblock Secure control for cyber-physical systems under malicious attacks.
\newblock {\em IEEE Transactions on Control of Network Systems}, 9(2):775--788, 2021.

\bibitem{XH19}
Maochao Xu and Lei Hua.
\newblock Cybersecurity insurance: Modeling and pricing.
\newblock {\em North American Actuarial Journal}, 23(2):220--249, 2019.

\bibitem{ZV21}
Lijing Zhai and Kyriakos~G Vamvoudakis.
\newblock Data-based and secure switched cyber-physical systems.
\newblock {\em Systems \& Control Letters}, 148:104826, 2021.

\bibitem{ZCKC20}
Lin Zhang, Xin Chen, Fanxin Kong, and Alvaro~A Cardenas.
\newblock Real-time attack-recovery for cyber-physical systems using linear approximations.
\newblock In {\em Proceedings of IEEE Real-Time Systems Symposium (RTSS)}, pages 205--217, 2020.

\bibitem{ZDO14}
Rui Zhuang, Scott~A DeLoach, and Xinming Ou.
\newblock Towards a theory of moving target defense.
\newblock In {\em Proceedings of the ACM Workshop on Moving Target Defense}, pages 31--40, 2014.

\end{thebibliography}

\begin{wrapfigure}{l}{0.8in}
\includegraphics[width=1in,height=1.25in,clip,keepaspectratio]{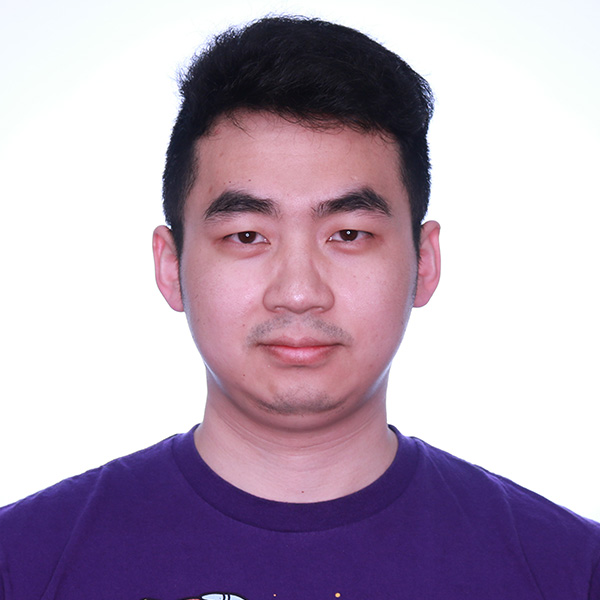} 
\end{wrapfigure}

\textbf{Hao Fu} was born in November 5, 1994, Anyang, China. He received his  Ph.D. in Department of Electrical and Computer Engineering in New York University, Tandon School of Engineering, Brooklyn, NY, USA, 2024. In 2019, he received his Master of Science degree in Electrical Engineering in the same department as well. He received his Bachelor of Science degree in Physics from University of Science and Technology of China, Hefei, China, 2017. His major field of study contains machine learning, finance, and control theory. From 2017 to 2018, he was a research assistant in NYU Wireless lab. In 2018 Fall, he joined in Control/Robotics Research Laboratory (CRRL). Previously, he was studying the possibility of using machine learning tools to develop economical navigation algorithms. Additionally, he was also studying the possibility of using neural networks to assist decision-making in finance. Currently, he is studying backdooring attacks against neural networks and security problems in cyber-physical systems. Dr. Fu published articles in many journals, including IEEE Transactions on Information Forensics and Security, IEEE Transactions on Dependable and Secure Computing, and IEEE Access.

\begin{wrapfigure}{l}{0.8in}
\includegraphics[width=1in,height=1.25in,clip,keepaspectratio]{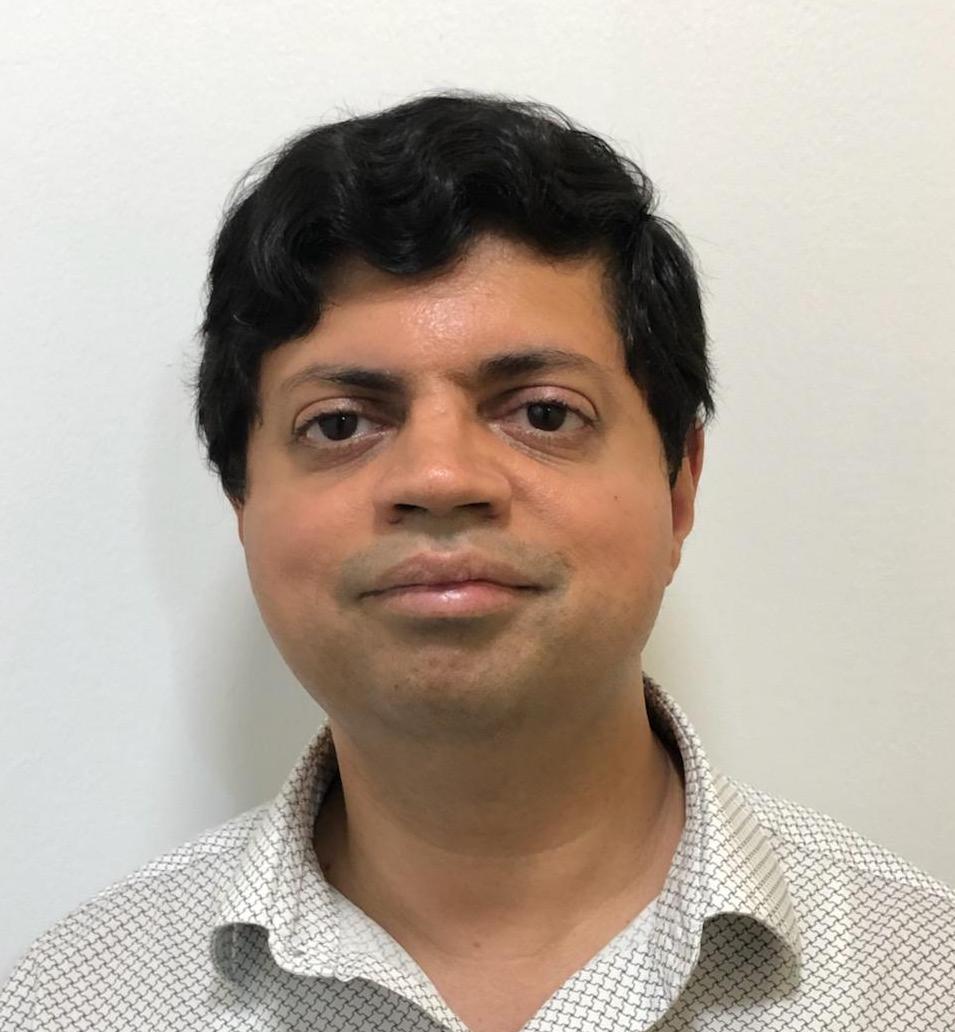} 
\end{wrapfigure}

\textbf{Prashanth Krishnamurthy}
received his B.Tech. degree in electrical engineering from Indian Institute of Technology, Chennai, India in 1999, and M.S. and Ph.D. degrees in electrical engineering from Polytechnic University (now NYU) in 2002 and 2006, respectively. He is a Research Scientist and Adjunct Faculty with the Department of Electrical and Computer Engineering at NYU Tandon School of Engineering. He has co-authored over 175 journal and conference papers and a book. He has also co-authored the book ``Modeling and Adaptive Nonlinear Control of Electric Motors'' published by Springer Verlag in 2003. His research interests include autonomous vehicles and robotic systems, multi-agent systems, sensor data fusion, robust adaptive nonlinear control, resilient control,  path planning and obstacle avoidance, machine learning, real-time embedded systems, cyber-physical systems and cyber-security, real-time software implementations, and decentralized and large-scale systems.   

\begin{wrapfigure}{l}{0.8in}
\includegraphics[width=1in,height=1.25in,clip,keepaspectratio]{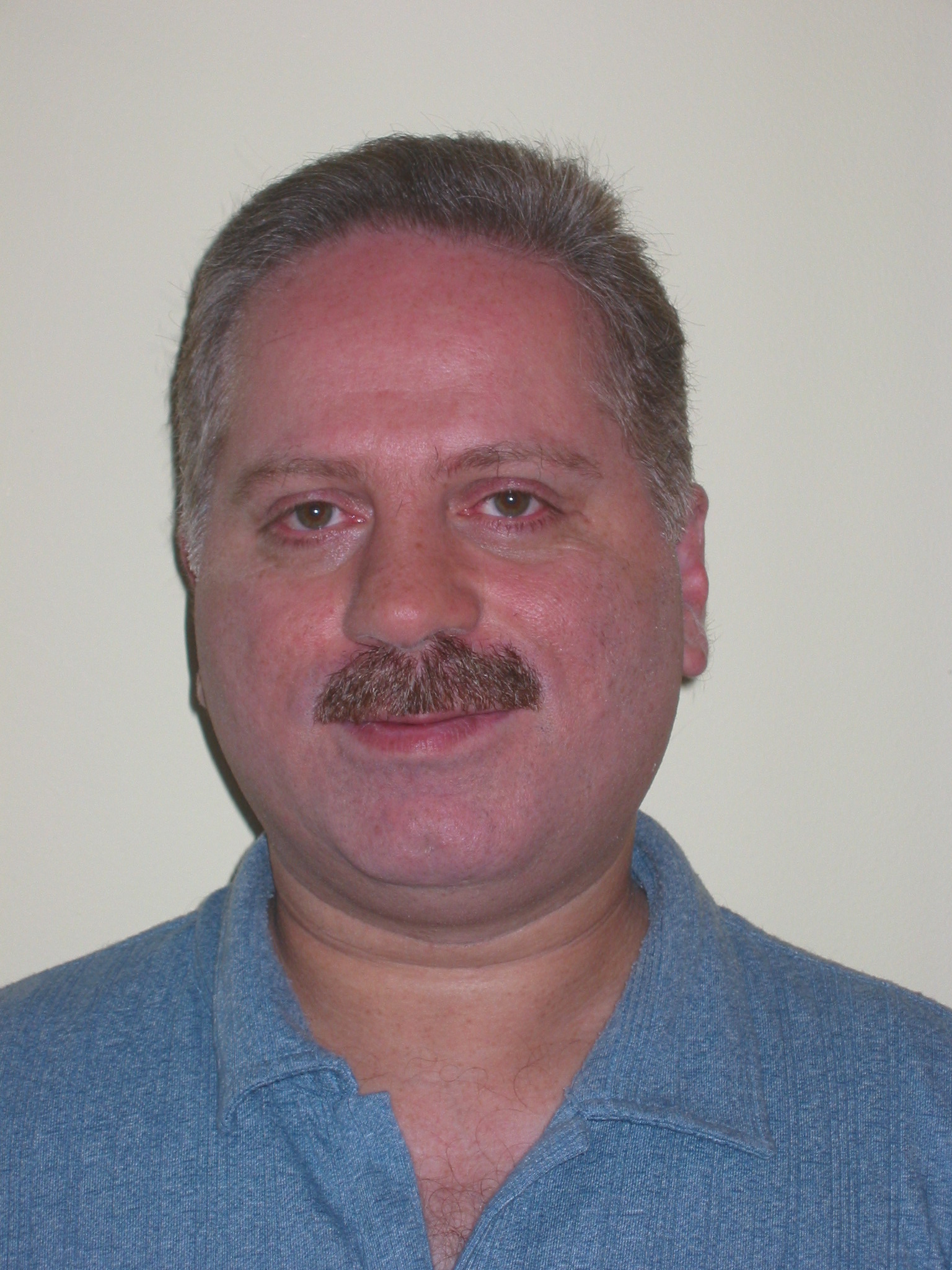} 
\end{wrapfigure}

\textbf{Farshad Khorrami} received the Bachelors degrees in mathematics and electrical engineering in 1982 and 1984 respectively from The Ohio State University. He also received the Master’s degree in mathematics and Ph.D. in electrical engineering in 1984 and 1988 from The Ohio State University, Columbus, Ohio, USA. He is currently a professor of Electrical and Computer Engineering Department at NYU, Brooklyn, NY where he joined as an assistant professor in Sept. 1988. His research interests include adaptive and nonlinear controls, robotics and automation, unmanned vehicles, cyber security for CPS, embedded systems security, machine learning, and large-scale systems and decentralized control. He has published over 360 refereed journal and conference papers in these areas. His book ``Modeling and Adaptive Nonlinear Control of Electric Motors'' was published by Springer  Verlag in 2003.  He also has fifteen U.S. patents on novel smart micro-positioners, control systems, cyber security, and wireless sensors and actuators. He has  developed and directed the Control/Robotics  Research  Laboratory at Polytechnic University (Now NYU) and the Co-Director of the Center in AI and Robotics (CAIR) at NYU Abu Dhabi. Dr. Khorrami has  also commercialized UAVs as well as development of  auto-pilots for various unmanned vehicles. His research has been supported by the ARO, NSF, ONR, DARPA, DOE,  AFRL, NASA, and several corporations. He has served as conference organizing committee member of several international conferences.

\end{document}